\documentclass[journal]{IEEEtran}
\usepackage{cite}
\usepackage{times,amsmath,amssymb,psfrag}
\usepackage{graphicx}
\usepackage{url}
\usepackage{bm}
\usepackage{multicol}
\graphicspath{{./pics/}} \DeclareGraphicsExtensions{.pdf,.jpg,.png}
\usepackage{mdframed}
\usepackage{xcolor}
\usepackage{hyperref}
\usepackage{epsfig}
\usepackage{float}
\usepackage{array}
\usepackage{diagbox} 
\usepackage{multirow}
\usepackage{etoolbox}
\usepackage{nomencl}
\usepackage{algorithm}
\usepackage{algorithmicx}
\usepackage{algpseudocode}
\usepackage{mathrsfs}
\usepackage{enumerate}
\usepackage{amsmath}
\usepackage{amsthm}
\usepackage{soul}

\soulregister\cite7 
\soulregister\citep7 
\soulregister\citet7 
\soulregister\ref7 
\soulregister\pageref7 

\newtheorem{definition}{Definition}
\newtheorem*{theorem}{Theorem}
\newtheorem{lemma}{Lemma}

\newtheorem{remark}{Remark}
\newtheorem{assumption}{Assumption}
\usepackage{booktabs}

\DeclareMathOperator{\sat}{sat}
\DeclareMathOperator{\VP}{VP}
\DeclareMathOperator{\DB}{DB}

\begin{document}

\title{High-Fidelity Large-Signal Order Reduction Approach for Composite Load Model}

\author{Zixiao~Ma,
	~Zhaoyu Wang,
	~Dongbo Zhao,
	and ~Bai Cui
	\thanks{
		Zixiao Ma and Zhaoyu Wang are with the Department of Electrical and Computer Engineering, Iowa State University, Ames, IA 50011, USA (email: zma@iastate.edu, wzy@iastate.edu).
		
		Dongbo Zhao is with Argonne National Laboratory, Argonne, IL 60439, USA (email: dongbo.zhao@anl.gov).
		
		Bai Cui is with National Renewable Energy Laboratory, Golden, CO 80401, USA (email: bcui@nrel.gov).
	}
}

\maketitle

\begin{abstract}
With the increasing penetration of electronic loads and distributed energy resources (DERs), conventional load models cannot capture their dynamics. Therefore, a new comprehensive composite load model is developed by Western Electricity Coordinating Council (WECC). However, this model is a complex high-order nonlinear system with multi-time-scale property, which poses challenges on stability analysis and computational burden in large-scale simulations. In order to reduce the computational burden while preserving the accuracy of the original model, this paper proposes a generic high-fidelity order reduction approach and then apply it to WECC composite load model. First, we develop a large-signal order reduction (LSOR) method using singular perturbation theory. In {this method}, the fast dynamics are integrated into the slow dynamics to preserve the transient characteristics of fast dynamics. Then, we propose the necessary conditions for accurate order reduction and embed them into the LSOR to improve and guarantee the accuracy of reduced-order model. Finally, we develop the reduced-order WECC composite load model using the proposed algorithm. Simulation results show the reduced-order large signal model significantly {alleviates} the computational burden while maintaining similar dynamic responses as the original composite load model.
\end{abstract}

\markboth{Submitted to IET for possible publication. Copyright may be transferred without notice}%
{Shell \MakeLowercase{\textit{et al.}}: Bare Demo of IEEEtran.cls for Journals}
\section{Introduction}\label{sec1}
\IEEEPARstart{P}{ower} system {load} modeling is important {in} stability analysis, optimization, and controller design [1]. Although this topic has been widely studied, it is still a challenging problem due to increasing diversity of load components and lack of detailed load information and measurements. 
	
Load models can be classified {into} static and dynamic {ones}. Static load models such as static constant impedance-current-power (ZIP) model and exponential model have simple model structures [2][3]. However, they cannot capture the dynamic load behaviors [4]-[10]. Motivated by the 1996 blackout of the Western Systems Coordinating Council (WSCC), a widely-used dynamic composite load model was developed [11]. The model consists of a ZIP and a dynamic induction motor (IM). It was designed to represent highly stressed {loading} conditions in summer peak {hours}. However, this interim load model was unable to {capture} the fault-induced delayed voltage recovery (FIDVR) events [7]. A preliminary WECC composite load model (WECC CLM) was proposed by adding an impedance representing the electrical distance between substation and end-users, an electronic load and a single-phase motor [12]-[14]. After a series of improvements, the latest WECC composite load model (CMPLDWG) is developed as shown in Fig. \ref{CMPLDWG}. The electrical distance between the substation and end-users is represented by a substation transformer, a shunt reactance, and a feeder equivalent. The model consists of three three-phase motors, one aggregate single-phase AC motor, one static load, one power electronics component, and {one} distributed energy resource (DER). The DER in CMPLDWG is currently represented by the PVD1 model [15]. However, PVD1 has 5 modules, 121 parameters, and 16 states, which is as complex as the CMPLDW itself. Therefore, the Electric Power Research Institute (EPRI) has developed a simpler yet more comprehensive model to replace PVD1, which is named as DER\_A model [15].
	
The above WECC CMPLDW + DER\_A model is a \textit{complex high-order nonlinear} dynamical system with \textit{multi-time-scale} property, which means the state vector is high-dimensional and the transient velocity of each state varies significantly. These characteristics result in two main challenges. Firstly, it increases the difficulty of dynamic stability analysis due to the numerous state variables.  Secondly, it makes simulation studies of a high-order power system computationally demanding or even infeasible. There are two main reasons for this high computational burden. One {reason is the shear dimensionality of the problem}. The other {comes from} the two-time-scale property of the model. This makes solving the model a stiff ordinary differential equation (ODE) problem, which requires small time steps to calculate the fast dynamics, and consequently results in long computational time to capture slow dynamics. The fast dynamics are often introduced by the intentionally added inductance and capacitance, moment of inertia, and parasitic elements inherent in the system [16]. However, simply neglecting the fast dynamics may lead to modeling inaccuracies in dynamic response and stability property. In order to accelerate computation while maintaining the accuracy and faithful stability property of the original load model, it is imperative to develop a high-fidelity reduced-order load model. To our best knowledge, this is the first paper on dynamic order reduction of WECC composite load model especially containing the DER\_A model.

The existing model reduction methods usually project the higher dimensional counterpart into a lower dimensional subspace where dynamic features of the original model dominate. Singular perturbation is the kind of method which considers the fast dynamics as boundary-layers and includes their solutions into slow dynamics. Singular perturbation method is suitable for analyzing two-time-scale problems and is widely used in power systems analysis. Previous applications include the derivation of reduced-order modeling of synchronous machines [17], microgrids [18], and distribution grid-tied systems with wind turbines [19]. However, these papers fall short of guaranteed accuracy and cannot be directly applied to the WECC composite load model due to the different system characteristics. 

Therefore, this paper develops a novel accuracy assessment theorem which takes into account the impact of external inputs on the accuracy of reduced system. By embedding the theorem, we proposes a high-fidelity order reduction approach for WECC composite load model. The derived high-fidelity reduced-order model can replace the original model in power system simulations for stability analysis and control applications with less computational complexity. Specifically, we improve the accuracy from two aspects. Firstly, without any simplification or linearization, we adopt the large-signal order reduction (LSOR) method based on singular perturbation theory to maintain all the dynamic characteristics of the original system. Secondly, we propose necessary conditions for accurate order reduction, and then integrate them into the LSOR method to theoretically guarantee the high accuracy of the reduced-order model. Note that this proposed approach is general and can be applied to various dynamic models.
	
The rest of the paper is organized as follows. Section \ref{C1} proposes high-fidelity order reduction approach in a general form. Section \ref{C2} introduces mathematical representation of WECC composite load model. Section \ref{C3} derives the reduced-order model using the proposed method. Section \ref{C4} shows the simulation results and analysis. Section \ref{C5} concludes the paper.
	\begin{figure}[t!]
		\centering
		\includegraphics[width=8.5cm]{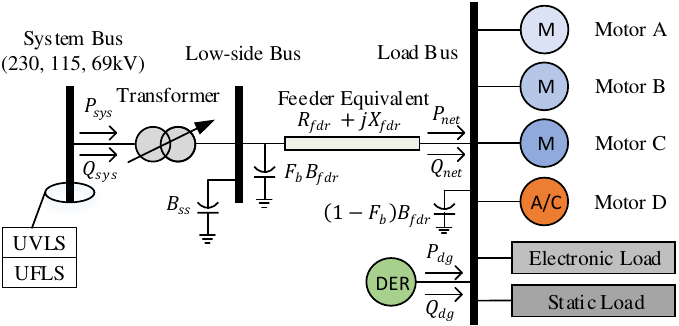}
		\caption{A schematic diagram of the WECC CMPLDWG [13].}
		\label{CMPLDWG}
	\end{figure}

\section{High-Fidelity Order Reduction Method}\label{C1}
Accurate load modeling is essential to power system stability analysis, optimization and control. To solve the challenges raised by high-order characteristics of WECC composite load model, we propose a general approach for high-fidelity order reduction in this section. We first introduce the LSOR method based on singular perturbation theory. A novel accuracy assessment theorem is then derived and embedded into the LSOR to guarantee the accuracy of reduced-order model.

\subsection{LSOR Based on Singular Perturbation Theory}
Consider a standard singular perturbation model as follows,
\begin{align}
    \dot {x} &= f\left( x,z,u,\varepsilon\right),   \label{standardx} \\
    \varepsilon\dot {z} &= g\left( x,z,u,\varepsilon\right), \label{standardz}
\end{align}
where $x\in \mathbb{R}^n$ represents slow state vector, $z\in \mathbb{R}^m$ denotes fast state vector, $u\in \mathbb{R}^p$ denote external input vector, and $\varepsilon\in[0,\varepsilon_0]$; $f$ and $g$ are Lipshitz continuous functions.
\begin{remark}
	Selecting the perturbation coefficient $\varepsilon$ for real physical systems is challenging. In most cases, we pick it based on our knowledge of the real system. In cases where it is unclear which parameter is small, we can locally linearize the system around the equilibrium point and use modal decomposition to identify the slow and fast dynamics. 
\end{remark}
When $\varepsilon$ is small, the fast transient velocity $\dot {z}=g/\varepsilon$ can be much larger than that of the slow transient $\dot{x}$. To solve this two-time-scale problem, we can set $\varepsilon=0$, then equation (\ref{standardz}) degenerates to the following algebraic equation,
\begin{eqnarray}\label{algebraic}
	0 = g\left( x,z,u,0\right).
\end{eqnarray}
Assuming that equation (\ref{algebraic}) has at least one isolated real root, and satisfies the implicit function theorem, then for each argument, we can obtain the quasi-steady-state (QSS) solution in a local vicinity around the isolated root,
\begin{eqnarray}\label{roots}
	z = h\left( x,u\right). 
\end{eqnarray}
Substituting equation (\ref{roots}) into equation (\ref{standardx}) and setting $\varepsilon=0$, we obtain the QSS model,
\begin{eqnarray}\label{quasi}
	\dot {x} = f\left( x,h\left( x,u\right),u,0\right).
\end{eqnarray}
We call the QSS system \eqref{quasi} the reduced-order model since its order drops from $n+m$ to $n$. The slow states can be obtained by solving the reduced-order model (\ref{quasi}), whereas the fast states are represented by equation (\ref{roots}). However, \eqref{roots} only gives approximate solution unless $\varepsilon$ is zero. To quantify the error between approximate and actual fast states, we denote the error as $y=z-h(x,u)$. Then in the fast-time-scale $\tau= t /\varepsilon$, the dynamics of $y$ are governed as follows,
\begin{align}
    \frac{\mathrm{d}y}{\mathrm{d}\tau}=&G(x,y,u,\varepsilon)\nonumber\\
    =&g(x,h(x,u)+y,u,\varepsilon)\nonumber\\
    &-\varepsilon\left[\frac{\partial h}{\partial x}f(x,h(x,u)+y,u,\varepsilon) + \frac{\partial h}{\partial u}\dot{u}\right].
\end{align}
Let $\varepsilon=0$, we obtain the boundary-layer model:
\begin{eqnarray}\label{blm}
	\frac{\mathrm{d}y}{\mathrm{d}\tau} = g\left( x,y+h\left( x,u\right),u,0\right). 
\end{eqnarray}

Note that the exact fast states are $z=y+h(x,u)$, but we do not know $(x,y)$. Therefore, if we can guarantee the accuracy of reduced-order model and boundary-layer model, then we can use their solutions $(\hat{x},\hat{y})$ instead of $(x,y)$. However, these models are exact only when $\varepsilon$ is exactly zero, which is obviously not the case for the studied system. Thus, we need to quantitatively assess the accuracy of reduced-order model when $\varepsilon$ is small yet nonzero. This motivates the next subsection.

\subsection{High-Fidelity LSOR with Accuracy Assessment}
Before deriving the performance guarantee of the proposed high-fidelity order reduction approach, we first introduce a few technical definitions and assumptions:

\begin{definition}
Class $\mathcal{K}$ function $\alpha :[0,t)\rightarrow [0,\infty )$ is a continuous strictly increasing function with $\alpha(0) = 0$.
\end{definition}
\begin{definition}
Class ${\mathcal {KL}}$ function $\beta :[0,t)\times [0,\infty )\rightarrow [0,\infty )$ is a continuous function satisfying: for each fixed $s$, the function   $\beta (r,s)$  belongs to class $\mathcal{K}$; for each fixed  $r$, the function $\beta (r,s)$ is decreasing with respect to $s$ and $\beta (r,s)\rightarrow 0$ for $s\rightarrow \infty$.
\end{definition}
\begin{definition}
$|f|=O(\varepsilon)$ is equivalent to $|f| \leqslant k \varepsilon$.
\end{definition}
\begin{assumption} \label{thm:as1}
	 The functions $f$, $g$, and their first partial derivatives are continuous and bounded with respect to $(x,z,u,\varepsilon)$; $h$ and its first partial derivatives $\partial h/\partial x$, $\partial h/\partial u$ are locally Lipschitz; and the Jacobian $\partial g/\partial z$ has bounded first partial derivatives with respect to its arguments.  
\end{assumption}
\begin{assumption} \label{thm:as2}
The reduced-order model (\ref{quasi}) is input-to-state stable with Lyapunov gain ${\alpha}$ as follows,
\begin{eqnarray}\label{ISS}
	\hat{x}\leqslant \beta(\|x(0)\|,t)+\alpha(\|u\|), 
\end{eqnarray}
where $\hat{x}$ is a solution of (\ref{quasi}), $\beta$ is a function of class $\mathcal{KL}$, $\alpha$ is a class $\mathcal{K}$ function, and $\|\cdot\|$ denotes any $p$-norm.
\end{assumption}

\begin{assumption} \label{thm:as3}
	The origin of the boundary-layer model (\ref{blm}) is a uniformly globally exponentially stable equilibrium and the solution $\hat{y}$ of (\ref{blm}) follows that
	\begin{eqnarray}\label{GES}
    \left\|\hat{y}(\tau)\right\|\leqslant k_1 e^{-a\tau},\;\forall\; \tau\geqslant0,
\end{eqnarray}
where $k_1$ and $a$ are positive constants.
\end{assumption}

Assumption \ref{thm:as1} describes the basic growth conditions on the original system which are commonly satisfied for power load models. Assumptions \ref{thm:as2} and \ref{thm:as3} are stability conditions on reduced-order model and boundary-layer model, respectively. 

Then, we propose the accuracy assessment index in the following accuracy assessment theorem, which will be embedded into the LSOR to realize high-fidelity order reduction. Before that, we give the following lemma for the proof of the theorem.

\begin{lemma}
Assume $\max\left\{ \left\|x(0) \right\|,\left\|y(0) \right\|,\left\| u \right\|,\left\| \dot{u} \right\|  \right\} \leqslant \mu$ holds for some positive constant $\mu$. Then there exists a class $\mathcal{KL}$ function $\beta_x$, a class $\mathcal{K}$ function $\alpha_x$ and positive constants $\mu_x$ and $\xi$ satisfying $\mu_x>\beta_x(\mu,0)+\alpha_x(\mu)+\xi$ such that $\|x(t)\|\leqslant \mu_x$ for all $t\in[0,\infty)$.
\end{lemma}

\begin{proof}\label{proof1}
According to the definition of class $\mathcal{KL}$ functions, we have $\beta_x(\mu,0)\leqslant\mu, \forall\; \mu\in\mathbb{R},$ so $\mu_x>\mu$. Since $x$ is continuous with finite initial conditions, we can find a maximal interval $[0,t_{\max})$, in which $\|x(t)\|\leqslant\mu_x$, where $t_{\max}>0$ is defined as the upper bound of the interval. From the definition of $\mu$ and the assumption that $t_{\max}$ is finite, there must be some positive constant $\Delta t$ such that $\|x(t)\|\leqslant\mu_x$ holds for all $t\in[0,t_{\max}+\Delta t)$. This contradicts that $t_{\max}$ is the upper bound, so $t_{\max}$ should be infinite and $\|x(t)\|\leqslant\mu_x, \forall \;t\in[0,\infty)$.
\end{proof}

\begin{theorem}
If Assumption \ref{thm:as1}--\ref{thm:as3} are satisfied, then there exist positive constants $\varepsilon^*$, $\mu$, such that for all  $t\in [0,\infty)$, $\max\left\lbrace\left\|x(0) \right\|,\left\|y(0) \right\|,\left\| u \right\|,\left\| \dot{u} \right\|  \right\rbrace \leqslant \mu$, and $\varepsilon \in (0,\varepsilon^*]$, the errors between solutions of original system (\ref{standardx})-(\ref{standardz}) and its reduced-order model (\ref{quasi}) and boundary-layer model (\ref{blm}) satisfy
\begin{align}
x(t,\varepsilon)-\hat{x}(t) =O( \varepsilon) , \label{slow} \\
z(t,\varepsilon)-h(\hat{x}(t),u(t))-\hat{y}(t/\varepsilon) =O( \varepsilon), \label{fast}
\end{align}
where $\hat{x}(t)$ and $\hat{y}(\tau)$ are the solutions of reduced-order model (\ref{quasi}) and boundary-layer model (\ref{blm}), respectively. Furthermore, for any given $T>0$, there exists a positive constant $\varepsilon^{**}\leqslant\varepsilon^*$ such that for $t\in\left[ T,\infty\right) $ and $\varepsilon<\varepsilon^{**}$, it follows that
\begin{eqnarray}\label{morefast}
z(t,\varepsilon)-h(\hat{x}(t),u(t))=O( \varepsilon).
\end{eqnarray}
\end{theorem}
\begin{remark}
 Equation (\ref{fast}) means that when Assumption 1 and 3 are satisfied, we can use $h+\hat{y}$ to accurately represent the solution of fast dynamics for $\varepsilon\in[0,\varepsilon^{*}]$ and bounded inputs. However, it requires solving the boundary-layer model. Further, (\ref{morefast}) means that if $\varepsilon\leqslant\varepsilon^{**}<\varepsilon^*$, the solution of fast transient can be estimated by only $h(t,\hat{x}(t))$ after $T>0$. {This result significantly simplifies the order reduction.}
\end{remark}
\begin{proof}
From Assumption 2, we know that the solution of reduced-order model is bounded for bounded inputs. Therefore, we can expect that $x$ is also bounded if $\|x-\hat{x}\|=O(\varepsilon)$. However, we cannot use this inequality since it has not been proven yet. Therefore, we exploit signal truncation as Lemma 1 to show that $x$ is in a compact set. Then by Assumption 1, we have that the argument of $f(x,z,u,\varepsilon)$ is compact. Since $f$ is continuous, it follows that $f$ is bounded, i.e., $|f|\leqslant k_0$, and $x(t)$ is Lipshitz.

Then using Assumption 1, 3 and Lemma 9.8 in [16], we conclude that there exists a Lyapunov function $V_y(x,y,u)$ and positive constants $b_1,b_2,\dots,b_6$ and $\rho_0$ satisfying
\begin{align}
    b_1\|y\|^2\leqslant V_y(x,y,u)\leqslant& \;b_2\|y\|^2,\label{981}\\
    \frac{\partial V_y}{\partial y}G(x,y,u,0)\leqslant-&b_3\|y\|^2,\label{982}\\
    \left\|\frac{\partial V_y}{\partial y}\right\|\!\leqslant\! b_4 \left\|y\right\|;\left\|\frac{\partial V_y}{\partial x}\right\|\!\leqslant\! b_5\left\|y\right\|^2;&\left\|\frac{\partial V_y}{\partial u}\right\|\!\leqslant\! b_6\left\|y\right\|^2,\label{983}
\end{align}
for all $y\in\{\|y\|<\rho_0\}$ and all $(x,u)\in\mathbb{R}^n\times \mathbb{R}^p$.

To assess the accuracy of solutions of fast dynamics, we define the estimation error as 
\begin{eqnarray}\label{error}
    \sigma_y(\tau,\varepsilon)=y(\tau,\varepsilon)-\hat{y}(\tau).
\end{eqnarray}
Differentiate both sides of equation (\ref{error}) and abbreviate $x(t_0+\varepsilon\tau,\varepsilon),y(\tau,\varepsilon),u(t_0+\varepsilon\tau,\varepsilon),\sigma_y(\tau,\varepsilon)$ as $x,y,u,\sigma_y$, respectively, then we have
\begin{align}\label{G}
    \frac{\partial\sigma_y}{\partial \tau}&=G(x,y,u,\varepsilon)-G(x_0,\hat{y},u_0,0)\nonumber\\
    &=G(x,\sigma_y,u,0)+\Delta G,
\end{align}
where $\Delta G=G(x,y,u,\varepsilon)-G(x,\sigma_y,u,0)-G(x_0,\hat{y},u_0,0)$. Utilizing the Lipshitz conditions of $G$ and $x$, and the condition of Lemma 1, we have
\begin{align}\label{deltag}
    &\|\Delta G\|\!\leqslant \!k_2\|\sigma_y\|^2\!+\!\varepsilon l_1\!+\!(k_3\|\sigma_y\|\!+\!l_2|u\!-\!u_0|\!+\!l_3\|x\!-\!x_0\|\!)\|\hat{y}\|\nonumber\\
    &\leqslant \!k_2\|\sigma_y\|^2\!+\!k_1\!k_3\|\sigma_y\|e^{-a\tau}\!+\!\varepsilon l_1\!+\!\varepsilon(l_2\mu \tau\!+\!l_3k_4\!+\!l_3k_0\tau)e^{-a\tau}\nonumber\\
    &\leqslant \!k_2\|\sigma_y\|^2\!+\!k_1\!k_3\|\sigma_y\|e^{-a\tau}\!+\!\varepsilon k_5,
\end{align}
for some nonnegative constants $a,k_i, i=1,\dots,5$ and nonnegative Lipshitz constants $l_j, j=1,\dots,3$, where $k_5=l_1+k_1 \max\{l_3k_4,l_2+l_3k_0\}\times \max\{1,1/a\}$.

Equation (\ref{G}) can be viewed as the perturbation of 
\begin{align}
    \frac{\partial\sigma_y}{\partial \tau}=G(x,\sigma_y,u,0).
\end{align}
Using (\ref{981})-(\ref{983}) and (\ref{deltag}), the derivative of Lyapunov funtion $V_y(x,\sigma,u)$ along the trajectories of (\ref{G}) can be calculated as
\begin{align}
    \dot{V}_y=&\frac{\partial V_y}{\partial x}f+\frac{1}{\varepsilon}\cdot \frac{\partial V_y}{\partial \sigma_y}(G+\Delta G)+\frac{\partial V_y}{\partial u}\dot{u}\nonumber\\
    \leqslant& b_5k_0\|\sigma_y\|^2-\frac{b_3}{\varepsilon}\|\sigma_y\|^2+b_6\mu\|\sigma_y\|^2\nonumber\\
    &+\frac{b_4}{\varepsilon}\|\sigma_y\|(k_2\|\sigma_y\|^2+k_1k_3\|\sigma_y\|e^{-a\tau}+\varepsilon k_5)\nonumber\\
    \leqslant&-\frac{b_3}{2\varepsilon}\|\sigma_y\|^2+\frac{b_4k_1k_3}{\varepsilon}e^{-a\tau}\|\sigma_y\|^2+b_4k_5\|\sigma_y\|\nonumber\\
    \leqslant&-\frac{2}{\varepsilon}(\xi_1-\xi_2e^{-a\tau})V_y+2\xi_3\sqrt{V_y},
\end{align}
for $0<\varepsilon\leqslant\varepsilon^*$ and $\|\sigma_y\|\leqslant b_3/(4b_4k_2)$, where $\xi_1=b_3/(4b_2)$, $\xi_2\!=\!b_4k_1k_3/\!(2b_1)$, $\xi_3\!=\!b_4k_5/(2\sqrt{b_1})$, and $\varepsilon^*\!=\!b_3/(b_5k_0\!+\!b_6\mu)$.

Let $W_y=\sqrt{V_y}$ and use the comparison lemma, we have
\begin{align}
    W_y(\tau)&\leqslant \phi(\tau,0)W_y(0)+\varepsilon \xi_3 \int_{0}^{\tau}\phi(\tau,s)\;ds,\\
    |\phi(\tau,s)|&=\left|e^{-\int_{s}^{\tau}(\xi_1-\xi_2e^{-a\upsilon})\;d\upsilon}\right|\leqslant \xi_4e^{-\bar{a}(\tau-s)},\label{phitaus}
\end{align}
for some positive constants $\xi_4$ and $\bar{a}$. Since $\sigma_y(0)=O(\varepsilon)$, it follows that $\sigma_y(\tau)=O(\varepsilon)$ for all $\tau \geqslant0$. Then we can conclude that (\ref{fast}) holds $\forall\varepsilon\leqslant\varepsilon^*$ and $\forall\; t\geqslant0$. 

Moreover, from (\ref{GES}), we have $e^{-at/\varepsilon}\leqslant\varepsilon$, $\forall\; at\geqslant\varepsilon \ln(1/\varepsilon)$, then the term $\hat{y}(t/\varepsilon)$ will be $O(\varepsilon)$ on $[T,\infty)$ for $\varepsilon\in[0,\varepsilon^{**}]$, where $(\varepsilon^{**},T)$ is a pair of solution of 
\begin{eqnarray}\label{epsilon**}
    \varepsilon \ln \left(\frac{1}{\varepsilon}\right)=aT.
\end{eqnarray}

Now we have proved the accuracy of the solutions of fast dynamics. To show the conditions for \textit{accurate} solutions of slow dynamics, we can define $\sigma_x(t,\varepsilon)=x(t,\varepsilon)-\hat{x}(t)$. Following the similar procedure as (\ref{981})-(\ref{phitaus}), it can be verified that if Assumption 1-3 are satisfied, then (\ref{slow}) holds for $\varepsilon\in[0,\varepsilon^*]$ and all $t\in[0,\infty)$.
\end{proof}
\begin{remark}
    Note that $\varepsilon^*$ is a function of the bound of input signals and it follows that
    \begin{align}
        \lim_{\mu\to 0}{\varepsilon^*}=\frac{b_3}{b_5k_0}\; and\; \lim_{\mu\to +\infty}{\varepsilon^*}=0.
    \end{align}
     This means when inputs are zero, the upper bound of $\varepsilon$ is equal to that of its autonomous system; while when the inputs are unbounded, $\varepsilon$ must be exactly zero to guarantee the accuracy of the reduced-order model. This result reflects the impact of external inputs on the accuracy of the reduced-order model.
\end{remark}

The overall algorithm of this proposed high-fidelity order reduction method can be concluded as Algorithm \ref{alg:HFOR}.
\begin{algorithm}
\caption{High-Fidelity Order Reduction}\label{alg:HFOR}
\begin{algorithmic}[1]
\State {Find the perturbation coefficients $\varepsilon$. Identify the states with $\varepsilon$ as fast states, while the others as slow states.}
\Procedure{Reduced model derivation}{}
    \State {Let $\varepsilon=0$, solve the algebraic equation (\ref{algebraic}) to obtain the isolated QSS solutions $z = h\left({x},u\right)$.}
    \State {Substitute $z$ into (\ref{standardx}) to obtain reduced-order model (\ref{quasi})}
    \State {Derive the boundary-layer model using equation (\ref{blm}).}
    \EndProcedure
\Procedure{Calculate the Bound of $\varepsilon$}{}
    \State {Calculate $\varepsilon^*=b_3/(b_5k_0\!+\!b_6\mu)$.}
    \State {Calculate $\varepsilon^{**}$ by solving equation (\ref{epsilon**}).}
\EndProcedure
\Procedure{Accuracy Assessment}{}
    \If {\texttt{$ \varepsilon\leqslant\varepsilon^{*}$ }}
        \If{$\varepsilon\leqslant\varepsilon^{**}$}
        \State {$z=h(\hat{x},u)$ is the solution of fast dynamics}
        \Else   \State Use $z=h(\hat{x},u)+\hat{y}$ by solving (\ref{blm}).
        \EndIf
    \Else  \State{Return to Step 1 to re-identify slow/fast dynamics}
    \EndIf
\EndProcedure
\end{algorithmic}
\end{algorithm}

\section{Mathematical Representation of WECC Composite Load Model}\label{C2}
To apply the singular perturbation theory, we need the mathematical representation of WECC composite load model, which can be found in our previous work [20]. Since our objective is to reduce the order of \textit{dynamic} parts, the static ones such as single-phase motor (which is modeled as a performance model [14]), electronic loads [2] and static load are out of the scope of this paper. For brevity, only the mathematical representation of dynamic components are introduced in this section.

\subsection{Three-Phase Motor Model}
WECC composite load model uses three three-phase fifth-order induction motors, called motor A, B and C, to represent different types of dynamic components. These three-phase motors have the same structure but different parameter settings. The block diagram of the induction motor model is shown in Fig. \ref{threephase}. There are four dynamic equations with respect to ${E}_{\rm q}'$, ${E}_{\rm d}'$, ${E}_{\rm q}{''}$ and ${E}_{\rm d}{''}$. By adding the dynamic equation governing the slip $s$, we can represent the complete fifth-order model as follows,
\begin{figure}[t!]
	\centering
	\includegraphics[width=8cm]{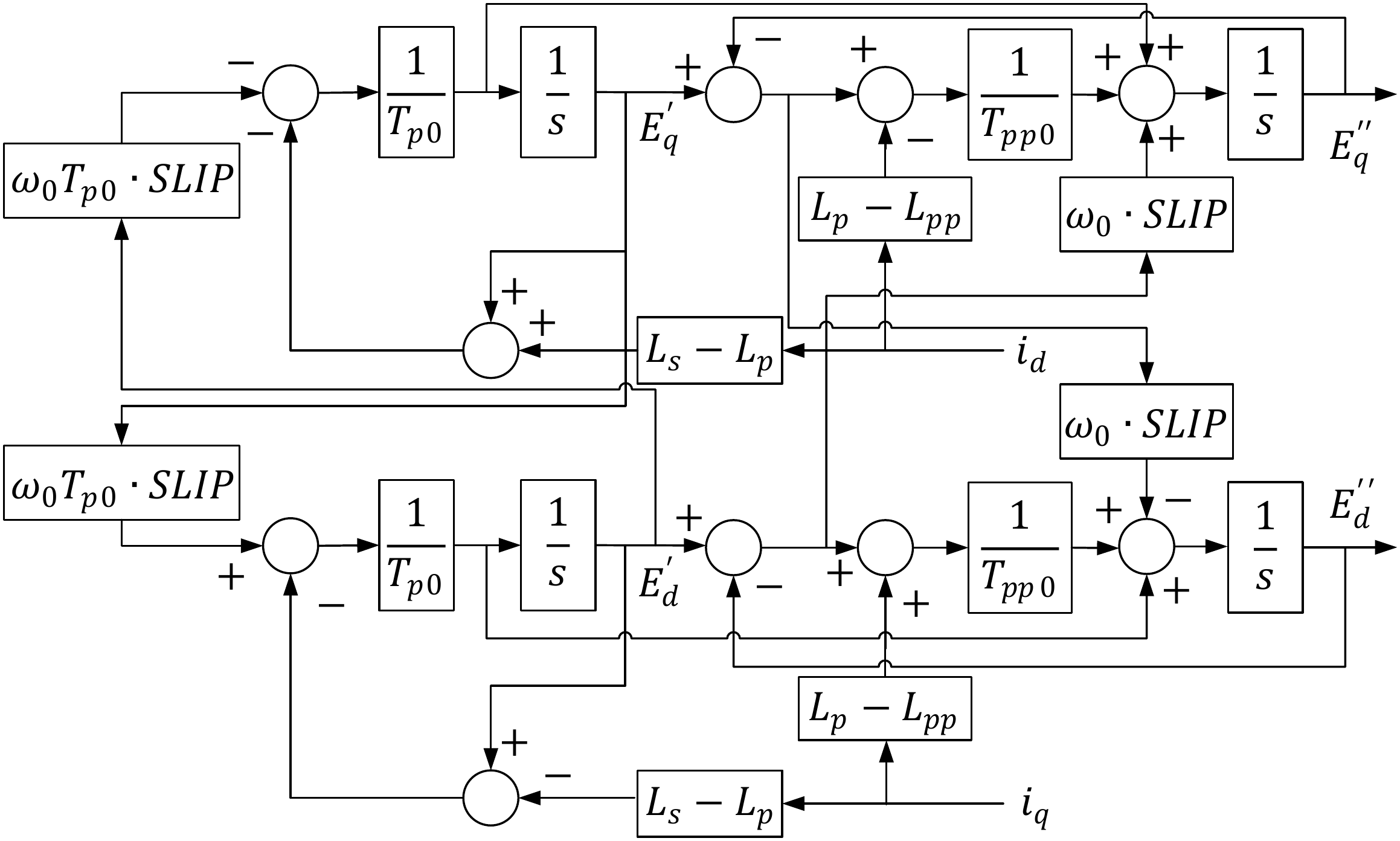}
	\caption{The block diagram of three-phase motor adopted in the WECC composite load model [14].}
	\label{threephase}
\end{figure}
\begin{align}
&\dot {E}_{\rm q}'\! = \!\frac{1}{{{T_{\rm p0}}}}\left[ { - E_{\rm q}'\! -\! {i_{\rm d}}\left( {{L_{\rm s}} \!- \!{L_{\rm p}}} \right)\! - \!E_{\rm d}'\! \cdot \!{\omega _0} \!\cdot\! s \!\cdot\! {T_{\rm P0}}} \right],\\\label{Eq'}
&\dot {E}_{\rm d}' \!=\! \frac{1}{{{T_{\rm p0}}}}\left[ {-E_{\rm d}'\! + \!{i_{\rm q}}\left( {{L_{\rm s}}\! -\! {L_{\rm p}}} \right) \!+ \!E_{\rm q}' \!\cdot\! {\omega _0} \!\cdot\! s\! \cdot\! {T_{\rm P0}}} \right],\\\label{Ed'}
&\dot E_{\rm q}{''} \!\!={\frac{{T_{\rm p0}\!-\!T_{\rm pp0}}}{T_{\rm p0}T_{\rm pp0}}}E_{\rm q}'\!-\!{\frac{T_{\rm pp0}\left({ L_{\rm s}\!\!-\!\!L_{\rm p}}\right) \!+\!T_{\rm p0}\left( {L_{\rm p}\!\!-\!\!L_{\rm pp}}\right) }{{T_{\rm p0}T_{\rm pp0}}}}i_{\rm d}\!\!\\\nonumber
&	-{\frac{1}{T_{\rm pp0}}}E_{\rm q}''-\omega_0\cdot s \cdot E_{\rm d}'',\\\label{Eq''}
&\dot E_{\rm q}{''} \!\!={\frac{{T_{\rm p0}\!-\!T_{\rm pp0}}}{T_{\rm p0}T_{\rm pp0}}}E_{\rm q}'\!-\!{\frac{T_{\rm pp0}\left({ L_{\rm s}\!\!-\!\!L_{\rm p}}\right) \!+\!T_{\rm p0}\left( {L_{\rm p}\!\!-\!\!L_{\rm pp}}\right) }{{T_{\rm p0}T_{\rm pp0}}}}i_{\rm d}\!\!\\\nonumber
&	-{\frac{1}{T_{\rm pp0}}}E_{\rm q}''-\omega_0\cdot s \cdot E_{\rm d}'',\\\label{Ed''}
&\dot s =  - \frac{{p \cdot E_{\rm d}{''} \cdot {i_{\rm d}} + q \cdot E_{\rm q}{''} \cdot {i_{\rm q}} -T_{\rm L}}}{{2H}}.
\end{align}
The algebraic equations are:
\begin{align}
	T_{\rm L} &= {T_{\rm m0}}\left( {A{w^2} + Bw + {C_0} + D{w^{\rm Etrq}}} \right),\\
	{T_{\rm m0}}& = pE_{\rm d0}{''}{i_{\rm d0}} + qE_{\rm q0}{''}{i_{\rm q0}},\\
	w &= 1 - s,\\
	i_{\rm d}&={\frac{r_{\rm s}}{r_{\rm s}^2+L_{\rm pp}^2}}({V_{\rm d}+E_{\rm d}{''}} )+ {\frac{L_{\rm pp}}{r_{\rm s}^2+L_{\rm pp}^2}}({V_{\rm q}+E_{\rm q}{''}} ), \\
	i_{\rm q}&={\frac{r_{\rm s}}{r_{\rm s}^2+L_{\rm pp}^2}}({V_{\rm q}+E_{\rm q}{''}} )- {\frac{L_{\rm pp}}{r_{\rm s}^2+L_{\rm pp}^2}}({V_{\rm d}+E_{\rm d}{''}}), \\
	P&=V_{\rm d}i_{\rm d}+V_{\rm q}i_{\rm q}, \\
	Q&=V_{\rm q}i_{\rm d}-V_{\rm d}i_{\rm q}, 
\end{align}
where ${E}_{\rm q}'$, ${E}_{\rm d}'$, ${E}_{\rm q}{''}$, ${E}_{\rm d}{''}$ and $s$ are the five state variables; $L_{\rm s}$, $L_{\rm p}$ and $L_{\rm pp}$ are synchronous reactance, transient and subtransient reactance, respectively; $T_{\rm p0}$ and $T_{\rm pp0}$ are transient and subtransient rotor time constants, respectively; and $\omega_0$ is the synchronous frequency.
\subsection{DER\_A Model}
Recently, EPRI developed a new model to represent aggregated renewable energy resources named DER\_A which has fewer states and parameters than the previous PVD1 model. The dynamic model of DER\_A is as follows,
\begin{figure*}[ht!]
	\centering
	\includegraphics[width=18cm]{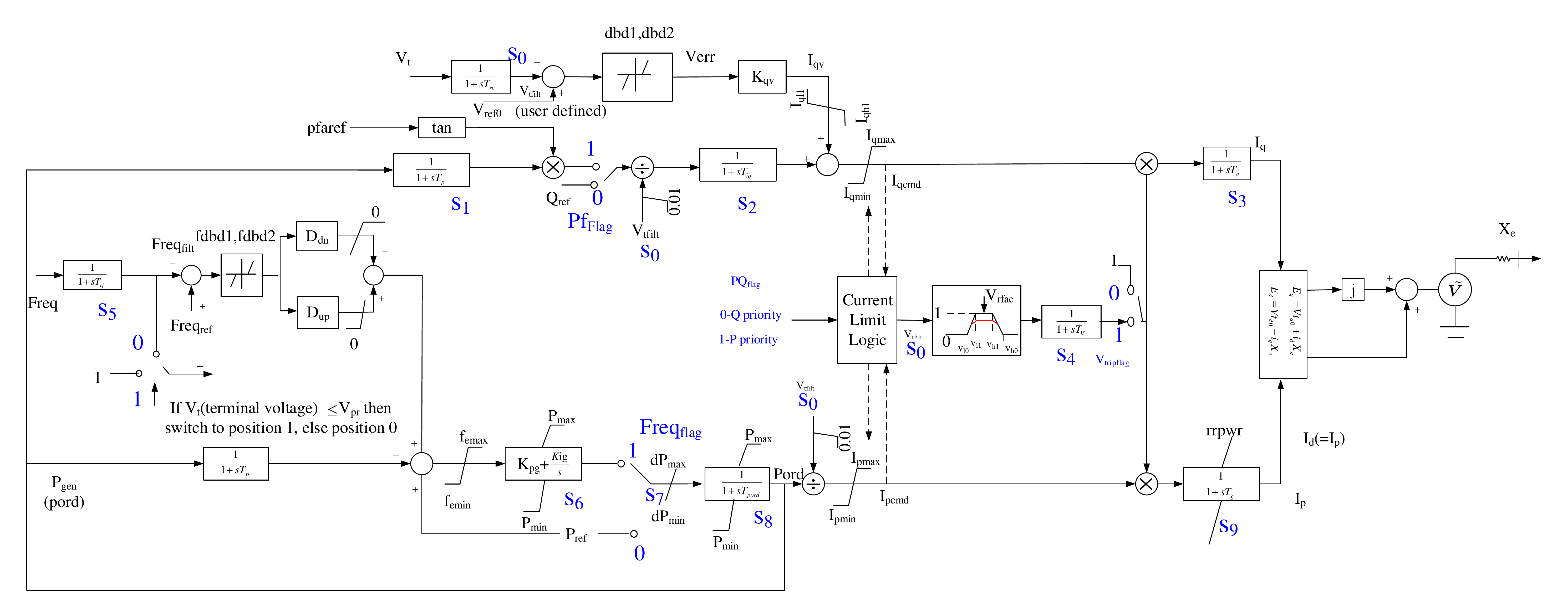}
	\caption{The block diagram of DER\_A in the WECC composite load model developed by [15].}.
	\label{DERA}
\end{figure*}
\begin{align}
	{\dot S_0} &= \frac{1}{{{T_{\rm rv}}}}\left( {{V_{\rm t}} - {S_0}} \right),\\
	{\dot S_1} &= \frac{1}{{{T_{\rm p}}}}\left( {{S_{8}} - {S_1}} \right),\\
	{\dot S_2} &= \!\!\left\{ \begin{aligned}
	-\frac{{S_2}}{{T_{\rm iq}}}\!  + \!\frac{{{Q_{{\mathrm{ref}}}}}}{{T_{\rm iq}}\cdot{\sat_1\left( {{S_0}} \right)}} {\kern 20pt} \text{if}{\kern 1pt} {\kern 1pt} {\kern 1pt} {{\rm Pf}_{{\mathrm{flag}}}} = 0,\\
	-\frac{S_2}{{T_{\rm iq}}} \!+ \!\frac{{\tan \left( {\rm pfaref} \right) \!\times\! {S_1}}}{{T_{\rm iq}}\cdot{\sat_1\left( {{S_0}} \right)}} {\kern 4pt} \text{if}{\kern 1pt} {\kern 1pt} {\kern 1pt} {\kern 1pt} {{\rm Pf}_{{\mathrm{flag}}}} = 1,
	\end{aligned} \right.\\
	{\dot S_3}&=\!\!\left\{\! \begin{aligned}
	\!\! \!\!\frac{\sat_2\!\left\lbrace  {S_2} \!+ \!\sat_3\left[  { \DB_{\rm V}\left( {V_{{\mathrm{ref0}}}} \!-\! {S_0} \right)\! \cdot\! {K_{\rm qv}}} \right]  \right\rbrace \!-\!{S_3}}{T_{\rm g}}\\
	{\kern 100pt}\text{if}{\kern 1pt} {\kern 1pt} {\kern 1pt} {V_{{\mathrm{tripflag}}}} = 0,\\
	\!\!\!\! \frac{ \sat_2\!\left\lbrace  {S_2}\! +\! \sat_3\!\left[ { \DB_{\rm V}\!\left( {{V_{{\mathrm{ref0}}}} \!- \!{S_0}} \right)\! \cdot \!{K_{\rm qv}}} \right]   \right\rbrace \!\! \cdot\!\! {S_4}\!-\!{S_3}}{T_{\rm g}}\\
	{\kern 130pt}\text{if}{\kern 1pt} {\kern 1pt} {\kern 1pt} {\kern 1pt} {V_{{\mathrm{tripflag}}}} = 1,
	\end{aligned} \right.\\\!\!\!\!\!\!\!\!\!\!\!\!
	{\dot S_4} &= \frac{1}{{{T_{\rm v}}}}\left( { \VP({S_0},{V_{{\mathrm{rfrac}}}}) - {S_4}} \right),\label{s4}\\
	{\dot S_5} &= \frac{1}{{{T_{\rm rf}}}}\left( {{\rm Freq} - {S_5}} \right),\label{s5}\\
	{\dot S_6} \!&= {K_{\rm ig}}\sat_4\{ {P_{{\mathrm{ref}}}} \!- \!{S_1}\! +\! \sat_5\left[ {{D_{\mathrm{dn}}} \!\cdot \!\DB_{\rm F}({\rm Freq}_{{\mathrm{ref}}}\! -\! {S_5})} \right]\nonumber \\
	&\;\;\;\;+\! \sat_6\left[ {{D_{\mathrm{up}}} \!\cdot\! \DB_{\rm F}({\rm Freq}_{{\mathrm{ref}}}\! -\! {S_5})} \right] \}\!+\! {\frac{K_{\rm pg}}{T_{\rm p}}}{S_1}\nonumber\\
	&\;\;\;\;+\!G_{\mathrm{dn}}\left( {{\rm Freq}-S_5}\right)+G_{\mathrm{up}}\left({{\rm Freq}-S_5}\right)  -\frac{S_{8}}{T_{\rm p}},\!\!\!\!\!\!\!\!\label{s6}\\
	{\dot S_7} &= \left\{ \begin{aligned}
	{0}{\kern 1pt} {\kern 1pt} {\kern 1pt} {\kern 1pt} {\kern 1pt} {\kern 1pt} {\kern 1pt} {\kern 1pt} {\kern 1pt} {\kern 1pt} {\kern 1pt} {\kern 1pt} {\kern 1pt} {\kern 1pt} {\kern 1pt} {\kern 1pt} {\kern 1pt} {\kern 1pt} {\kern 65pt} \text{if}{\kern 1pt} {\kern 1pt} {\kern 1pt} {\rm Freq}_{{\mathrm{flag}}} = 0,\\
	\sat_8\left[ \dot{\sat}_7\left( {{S_6}} \right)\right] {\kern 1pt} {\kern 1pt} {\kern 1pt} {\kern 1pt} {\kern 1pt} {\kern 1pt} {\kern 1pt} {\kern 1pt} {\kern 1pt} {\kern 1pt} {\kern 1pt} {\kern 1pt} {\kern 1pt} {\kern 1pt} {\kern 1pt} {\kern 1pt} {\kern 1pt} {\kern 1pt} {\kern 1pt} {\kern 1pt} {\kern 1pt} {\kern 1pt} {\kern 1pt} \text{if}{\kern 1pt} {\kern 1pt} {\kern 1pt} {\kern 1pt} {\rm Freq}_{{\mathrm{flag}}} = 1,
	\end{aligned} \right.\\
	{\dot S_8} &=\frac{1}{{{T_{\rm pord}}}}\left( {{S_7} - {S_8}} \right), {\kern 97pt}\\
	{\dot S_9} &=\! \left\{\! \begin{aligned}
	\!\frac{1}{T_{\rm g}}\!\left\lbrace \! {\sat_{9}\!\left[  \frac{{\sat_7({S_8})}}{{\sat_1\left( S_0 \right)}}\right]  \!\!\times\!\! {S_4}\!\! -\!\! {S_9}} \right\rbrace {\kern 2pt} \text{if} {\kern 2pt} {V_{{\mathrm{tripflag}}}}\! =\! 1,\\
	\frac{1}{T_{\rm g}}\!\left\lbrace \! {\sat_{9}\!\left[  \frac{{\sat_7({S_8})}}{{\sat_1\left( {{S_0}} \right)}}\right]  \!-\! {S_9}}\right\rbrace {\kern 12pt} \text{if}  {\kern 2pt} {V_{{\mathrm{tripflag}}}} \!=\! 0,
	\end{aligned} \right.\!\!\!\!
\end{align}
where $\sat_i(x), i = 1,\dots, 9$ are the saturation functions; $ \DB_{\rm V}(x)$ and $ \DB_{\rm F}(x)$ are deadzone functions with respect to voltage and frequency, respectively; and $ \VP(x,V_{{\mathrm{rfrac}}})$ represents the voltage protection function, which is a piece-wise algebraic function. The parameter definitions are given in Table \ref{table1}. Here we only summarize the dynamic equations that will be used in the order reduction. The complete detailed mathematical model can be found in [20].
	\renewcommand\arraystretch{1.2}
	\begin{table}
		\caption{Parameter definition of DER\_A model [15]}
		\label{table1}
		\setlength{\tabcolsep}{3pt}
		\begin{tabular}{|p{35pt}|p{200pt}|}
			\hline
			Parameters& 
			\centerline{Definitions} \\
			\hline
			$T_{\rm rv}$& 
			transducer time constant(s) for voltage measurement\\
			$T_{\rm p}$& 
			transducer time constant (s) \\
			$T_{\rm iq}$& 
			Q control time constant (s)\\
			$V_{\rm ref0}$& 
			voltage reference set-point $>$ 0 (pu)\\
			$K_{\rm qv}$& 
			proportional voltage control gain (pu/pu)\\
			$T_{\rm g}$& 
			current control time constant (s)\\
			${\rm Pf}_{\rm Flag}$& 
			0 $-$ constant Q control, and 1 $-$ constant power factor control\\
			$I_{\mathrm{max}}$& 
			maximum converter current (pu)\\
			$dbd1$& 
			lower voltage deadband $\leqslant 0$ (pu)\\
			$dbd2$& 
			upper voltage deadband $\geqslant 0$ (pu)\\
			$T_{\rm v}$& 
			time constant on the output of voltage/frequency cut-off\\
			$V_{\rm l0}$& 
			voltage break-point for low voltage cut-out of inverters\\
			$V_{\rm l1}$& 
			voltage break-point for low voltage cut-out of inverters\\
			$V_{\rm h0}$& 
			voltage break-point for high voltage cut-out of inverters\\
			$V_{\rm h1}$& 
			voltage break-point for high voltage cut-out of inverters\\
			$t_{\rm vl0}$& 
			timer for $V_{\rm l0}$ point\\
			$t_{\rm vl1}$& 
			timer for $V_{\rm l1}$ point\\
			$t_{\rm vh0}$& 
			timer for $V_{\rm h0} $point\\
			$t_{\rm vh1}$& 
			timer for $ V_{\rm h1}$ point\\
			$V_{{\mathrm{rfrac}}}$& 
			fraction of device that recovers after voltage comes back to within $V_{\rm l1} < V < V_{\rm h1}$\\
			$T_{\rm rf}$& 
			transducer time constant(s) for frequency measurement (must be $\geqslant 0.02 s$) \\
			$K_{\rm pg}$& 
			active power control proportional gain\\
			$K_{\rm ig}$& 
			active power control integral gain\\
			$D_{\mathrm{dn}}$& 
			frequency control droop gain $\geqslant 0$ (down-side)\\
			$D_{\mathrm{up}}$& 
			frequency control droop gain $\geqslant 0$ (up-side)\\
			$f_{\mathrm{emax}}$& 
			frequency control maximum error $\geqslant 0$ (pu)\\
			$f_{\mathrm{emin}}$& 
			frequency control minimum error $\leqslant 0$ (pu)\\
			$f_{\rm dbd1}$& 
			lower frequency control deadband $\leqslant 0$ (pu)\\
			$f_{\rm dbd2}$& 
			upper frequency control deadband $\geqslant 0$ (pu)\\
			${\rm Freq}_{\rm flag}$& 
			0 $-$ frequency control disabled, and 1 $-$ enabled\\
			$P_{\mathrm{min}}$& 
			minimum power (pu)\\
			$P_{\mathrm{max}}$& 
			maximum power (pu)\\
			$T_{\rm pord}$& 
			power order time constant (s)\\
			$dP_{\mathrm{min}}$& 
			power ramp rate down $<$ 0 (pu/s)\\
			$dP_{\mathrm{max}}$& 
			power ramp rate up $>$ 0 (pu/s)\\
			$V_{{\mathrm{tripflag}}}$& 
			0 $-$ voltage tripping disabled, 1 $-$ enabled\\
			$I_{\rm ql1}$& 
			minimum limit of reactive current injection, p.u.\\
			$I_{\rm qh1}$& 
			maximum limit of reactive current injection, p.u.\\
			$X_{\rm e}$& 
			source impedance reactive $>$ 0 (pu)\\
			$F_{{\mathrm{tripflag}}}$& 
			0 $-$ frequency tripping disabled, 1 $-$ enabled\\
			${\rm PQ}_{\rm flag}$& 
			0 $-$ Q priority, 1 $-$ P priority $-$ current limit\\
			${\rm typeflag}$& 
			0 $-$ the unit is a generator $I_{\rm pmin} = 0$, 1 $-$ the unit is a storage device and $I_{\rm pmin} = -I_{\rm pmax}$\\
			$V_{\rm pr}$& 
			voltage below which frequency tripping is disabled\\
			\hline
		\end{tabular}
		\label{tab1}
		\vspace {-0.5 em}
	\end{table}
\section{Reduced Order WECC Composite Load Model}\label{C3}
In this section, we will derive the reduced-order large-signal model of WECC composite load model using singular perturbation method. For the purpose of {order reduction}, we only focus on the dynamic components. These components are connected in parallel and we will reduce each individual component's order.

\subsection{Reduced-Order Three-Phase Motors Model} 

Each three-phase motor model has five states, $\tilde{x}_{\rm M}=\left[E_{\rm q}', E_{\rm d}', E_{\rm q}'', \right.\\\left.E_{\rm d}'', s \right].$ When applying the Algorithm 1, the first step is to identify the slow and fast dynamics. Since the fast dynamics are characterized by the small perturbation coefficient $\varepsilon$, we rewrite the left-hand-side of the dynamic equations as
\begin{eqnarray}\label{dynamics_nop}
	\left[T_{\rm p0}\dot{E}_{\rm q}',\  T_{\rm p0}\dot{E}_{\rm d}', \  T_{\rm pp0}\dot{E}_{\rm q}'',\  T_{\rm pp0}\dot{E}_{\rm d}'',\  H \dot s \right]^T.
\end{eqnarray}

\renewcommand\arraystretch{1.2}
\begin{table}
	\caption{Parameter setting of three-phase motor model}
	\label{table2}
	\setlength{\tabcolsep}{8pt}
	\centering
	\begin{tabular}{|c|c|c|c|c|c|}	
	    \hline
		\multicolumn{2}{|c|}{Motor A}&\multicolumn{2}{c|}{Motor B}&\multicolumn{2}{c|}{Motor C}\\
		\hline
		$r{\rm sA}$ & 
			0.04&
			$r{\rm sB}$& 
			0.03&
			$r{\rm sC}$& 
			0.03\\
			$L{\rm sA}$& 
			1.8&
			$L{\rm sB}$& 
			1.8&
			$L{\rm sC}$& 
			1.8\\
			$L{\rm pA}$& 
			0.1 &
			$L{\rm pB}$& 
			0.16&
			$L{\rm pC}$& 
			0.16\\
			$L{\rm ppA}$& 
			0.083&
			$L{\rm ppB}$& 
			0.12&
			$L{\rm ppC}$& 
			0.12 \\
			$T_{\rm poA}$& 
			0.092&
			$T_{\rm poB}$& 
			0.1&
			$T_{\rm poC}$& 
			0.1\\
			$T_{\rm ppoA}$& 
			0.002&
			$T_{\rm ppoB}$& 
			0.0026&
			$T_{\rm ppoC}$& 
			0.0026\\
			$H_{\rm A}$& 
			0.05&
			$H_{\rm B}$& 
			1&
			$H_{\rm C}$& 
			0.1\\
			$A_{\rm A}$& 	0&	$A_{\rm B}$& 	0&
			$A_{\rm C}$& 	0\\
			$B_{\rm A}$& 	0&
			$B_{\rm B}$& 	0&$B_{\rm C}$	&0 \\
			$C_{\rm A}$& 	0&	$C_{\rm B}$& 	0&
			$C_{\rm C}$& 	0\\
			$D_{\rm A}$& 	1&	$D_{\rm B}$& 	1&
			$D_{\rm C}$& 	1\\
			$E_{\rm trqA}$& 
			0 &
			$E_{\rm trqB}$& 
			2&
			$E_{\rm trqC}$& 
			2\\
			$p_{\rm A}$& 	-1&	$p_{\rm B}$& 	-1&
			$p_{\rm C}$& 	-1\\
			$q_{\rm A}$& 	-1&	$q_{\rm B}$& 	-1&
			$q_{\rm C}$& 	-1\\
			$\omega_{\rm 0A}$& 
			120$\pi$ &
			$\omega_{\rm 0B}$& 
			120$\pi$&
			$\omega_{\rm 0C}$& 
			120$\pi$\\
			\hline
		\end{tabular}
		\label{tab2}
	\end{table}  
Given one set of parameter setting in Table \ref{table2}, equation (\ref{dynamics_nop}) becomes
\begin{eqnarray}\label{dynamics_wp}
    \left[0.1\dot{E}_{\rm q}', 0.1\dot{E}_{\rm d}',  0.0026\dot{E}_{\rm q}'', 0.0026\dot{E}_{\rm d}'', 0.1 \dot s \right]^T.
\end{eqnarray}
The smaller perturbation coefficients in equation (\ref{dynamics_wp}) suggest that dynamic response velocities of $\left[E_{\rm q}', E_{\rm d}', s \right]^T$ are much slower than the rest of the states.  This difference is also an evidence of the two-time-scale property of this model. Then the slow and fast dynamics are divided as $\dot{\bar{x}}_{\rm M}=\left[ \dot{x}_{\rm M}, \dot{z}_{\rm M}\right]^T$, where $x_{\rm M}=\left[E_{\rm q}', E_{\rm d}', s \right]^T$, $z_{\rm M}=\left[E_{\rm q}'', E_{\rm d}''\right]^T$. For {consistency}, {denote} the input voltages $\left[V_{\rm q},\  V_{\rm d} \right]^T $ as $U_{\rm M}$. Following the singular perturbation method \eqref{standardx}--\eqref{quasi}, we can obtain the reduced-order large-signal model of three-phase motor as       
\begin{align}
   \!\!\!\!\!\! &\dot {x}_{\rm M1}\! = \!\frac{1}{{{T_{\rm p0}}}}\!\left[ { - x_{\rm M1}\! -\! {i_{\rm d}}\!\left( \!{{L_{\rm s}} \!\!-\! \!{L_{\rm p}}} \!\right)\! - \omega_0 \,{T_{\rm P0}\, x_{\rm M2} \, x_{\rm M3}}} \right]\!,\\
   \!\!\!\!\!\! &\dot {x}_{\rm M2}\! = \!\frac{1}{{{T_{\rm p0}}}}\!\left[ { - x_{\rm M2}\! +\! {i_{\rm q}}\!\left(\! {{L_{\rm s}} \!\!-\! \!{L_{\rm p}}}\! \right)\! + \omega_0 \,{T_{\rm P0} \, x_{\rm M1}  \,x_{\rm M3}}} \right]\!,\\
    \!\!\!\!\!\!&\dot {x}_{\rm M3}=\frac{T_{\rm L}-p\cdot h_2\left( x_{\rm M}\right) \cdot i_{\rm d}-q\cdot h_1\left( x_{\rm M}\right)\cdot i_{\rm q}}{2H}\!,
\end{align}
where the QSS solutions are 
\begin{align}
    h_1\!\left( x_{\rm M}\right)\!=&\frac{1}{r_{\rm s}^{2}\!+\!L_{\rm p}^{2}}\Big[\left(L_{\rm p} L_{\rm pp}\!+\!r_{\rm s}^{2}\right)\!x_{\rm M1}\!-\!\left(L_{\rm p}\! -\!L_{\rm pp}\right)\!r_{\rm s} x_{\rm M2} \nonumber\\ 
    \;& -\left(L_{\rm p} -L_{\rm pp}\right)L_{\rm p} U_1-\left(L_{\rm p} -L_{\rm pp}\right)r_{\rm s} U_2\Big],\\
    h_2\!\left( x_{\rm M}\right)\!=&\frac{1}{r_{\rm s}^{2}\!+\!L_{\rm p}^{2}}\Big[ \left(L_{\rm p}\! \;-\!L_{\rm pp}\right)\!r_{\rm s}x_{\rm M1}\!-\!\left(L_{\rm p} L_{\rm pp}\!+\!r_{\rm s}^{2}\right) x_{\rm M2} \nonumber\\
    & +\left(L_{\rm p}\! -\!L_{\rm pp}\right)\!r_{\rm s} U_1-\left(L_{\rm p} -L_{\rm pp}\right)L_{\rm p} U_2\Big].
\end{align}
Denote $\bar{x}_{\rm M3}=1-x_{\rm M3}$. {The} other algebraic equations are 
\begin{align}
    T_{\rm L} &= T_{\rm m0} \Big[ A\;\bar{x}_{\rm M3}^2 + B\; \bar{x}_{\rm M3} + {C_0} + D\;{ \bar{x}_{\rm M3}^{\rm Etrq}} \Big], \\ 
    T_{\rm m0} &= p\cdot h_2 \left( x_{\rm M}\right)\cdot i_{\rm d} + q\cdot h_1 \left( x_{\rm M}\right)\cdot i_{\rm q},\\
    i_{\rm q} &= \frac{r_{\rm s}}{r_{\rm s}^{2} \!+\! L_{\rm p}^{2}}\left(U_{\rm M1}\!+\!x_{\rm M1} \right) \!-\!\frac{L_{\rm p}}{r_{\rm s}^{2} \!+\! L_{\rm p}^{2}}\left(U_{\rm M2}\!+\!x_{\rm M2}\right),\\
    i_{\rm d} &= \frac{L_{\rm p}}{r_{\rm s}^{2}\! + \!L_{\rm p}^{2}}\left(U_{\rm M1}\!+\!x_{\rm M1}\right) \!+\!\frac{r_{\rm s}}{r_{\rm s}^{2} \!+\! L_{\rm p}^{2}}\left(U_{\rm M2}\!+\!x_{\rm M2}\right).
\end{align}

By solving equation (\ref{epsilon**}), we can find a pair of solution $(T,\varepsilon^{**})=(0.012,0.035)$. Since $\varepsilon_{\rm M}=0.0026<0.035$, the solutions of fast dynamics $\hat{z}_{\rm M}$ converge to $h(\hat{x}_{\rm M})$ exponentially fast within time $0.012\; sec$ which is short enough. Therefore, we can use only the QSS solution $h(\hat{x}_{\rm M})$ to represent the solution of fast dynamics.
\subsection{Reduced-Order DER\_A model}              
The DER\_A model has 10 states in total, {$\tilde{x}_{\rm D} = \left[S_0, S_1, \dots, S_9\right]^T$}. Different from the three-phase motor model, due to the existence of switches such as ${\rm Pf}_{\mathrm{flag}}$ and ${\rm PQ}_{\mathrm{flag}}$, the DER\_A model is actually a switching system consisting of {$2^6 = 64$} subsystems. Each subsystem is determined when the switches are fixed. Since these switches are preset, we only need to derive the reduced-order model for each subsystem. For brevity, we give the reduced-order model for one of the subsystems to illustrate the {model order reduction} procedure. The reduced-order models for other subsystems can be obtained using the same method. 
    
To find $\varepsilon$, we rewrite the dynamics as
\begin{multline}\label{dynamics_der}
     \left[T_{\rm rv}\dot{S}_0,\; T_{\rm p}\dot{S}_1,\; T_{\rm iq}\dot{S}_2,\; T_{\rm g}\dot{S}_3,\; T_{\rm v}\dot{S}_4,\; T_{\rm rf}\dot{S}_5,\right.  \\\left. T_{\rm p}\cdot T_{\rm rf}\dot{S}_6,\; \dot{S}_7,\; T_{\rm pord}\dot{S}_8,\; T_{\rm g}\dot{S}_9 \right]^T.
\end{multline}
Given the parameter setting in Table \ref{table3}, equation (\ref{dynamics_der}) becomes  
\begin{multline} \label{dynamics_der1}
    \left[0.1\dot{S}_0,\; 0.1\dot{S}_1,\; 0.005\dot{S}_2,\; 0.005\dot{S}_3,\; 0.005\dot{S}_4,\right.  \\ 
    0.1\dot{S}_5,\;\left. 0.1\cdot 0.1\dot{S}_6,\; \dot{S}_7,\; 0.005\dot{S}_8,\; 0.005\dot{S}_9 \right]^T.
\end{multline}
The smaller perturbation coefficients in equation (\ref{dynamics_der1}) suggest that dynamic response velocities of $\left[S_0, S_1, S_5, S_7\right]^T$ are much slower than other states. This difference is also an evidence of the {two-time-scale} property of this model. Then the slow and fast dynamics are divided as $\dot{\bar{x}}_{\rm D}=\left[\dot{x}_{\rm D}\dot{z}_{\rm D}\right]^T$, where $x_{\rm D}=\left[S_0,S_1,S_5,S_7 \right]^T$, $z_{\rm D}=\left[S_2,S_3,S_4,S_6,S_8,S_9\right]^T$. Defining the terminal voltage and frequency $\left[V_{\rm t},\  {\rm Freq} \right] $ as $U_{\rm D}$. {Following} the same procedure as above {(equations \eqref{standardx}--\eqref{quasi})}, we can derive the reduced-order large-signal model of DER\_A as       
    \renewcommand\arraystretch{1.1}
    \begin{table}
    \caption{Parameter setting of DER\_A model (Base: 12.47 kV and 15.0 MVA)}
    \label{table3}
    \setlength{\tabcolsep}{10pt}
    \centering
    \begin{tabular}{|c|c|c|c|}	
    	\hline
    	Parameters                  & Values        & Parameters        & Values \\
    	\hline
    	$T_{\rm rv}$                & 0.1 s         & $T_{\rm p}$       & 0.1 s \\
    	$T_{\rm iq}$                & 0.005 s       & $V_{\rm ref0}$    & 0 pu \\
    	$K_{\rm qv}$                & 5 pu/pu       & $T_{\rm g}$       & 0.005 s \\
    	${\rm Pf}_\mathrm{flag}$    & 1             & $I_\mathrm{max}$  & 1.2 pu \\
    	${\rm dbd1}$                & -0.05 pu      & ${\rm dbd2}$      & 0.05 pu \\
    	$T_{\rm v}$                 & 0.005 s       & $V_{\rm l0}$      & 0.44 pu \\
    	$V_{\rm l1}$                & 0.49 pu       & $V_{\rm h0}$      & 1.2 pu \\
    	$V_{\rm h1}$                & 1.15 pu       & $t_{\rm vl0}$     & 0.16 s \\
    	$t_{\rm vl1}$               & 0.16 s        & $t_{\rm vh0}$     & 0.16 s \\
    	$t_{\rm vh1}$               & 0.16 s        & $V_\mathrm{rfrac}$& 0.7 \\
        $T_{\rm rf}$                & 0.1 s         & $K_{\rm pg}$      & 0.1 pu \\
    	$K_{\rm ig}$                & 10 pu         & $D_\mathrm{dn}$   & 20 pu \\
    	$D_{\mathrm{up}}$           & 0 pu          & $f_\mathrm{emax}$ & 99 pu \\
    	$f_{\mathrm{emin}}$         & -99 pu        & $f_{\rm dbd1}$    & -0.0006 \\
    	$f_{\rm dbd2}$              & 0.0006        & ${\rm Freq}_\mathrm{flag}$    
    	& 0 \\
    	$P_{\mathrm{min}}$          & 0 pu          & $P_{\max}$        & 1.1 pu \\
    	$T_{\rm pord}$              & 0.005 s       & $dP_{\min}$       & -0.5 pu/s \\
        $dP_{\max}$                 & 0.5 pu/s      &$V_\mathrm{tripflag}$  
        & 1 \\
    	$I_{\rm ql1}$               & -1 pu         & $I_{\rm qh1}$     & 1 pu \\
    	$X_{\rm e}$                 & 0.25 pu       & $F_\mathrm{tripflag}$ 
    	& 1 \\
    	${\rm PQ}_\mathrm{flag}$    & 0             & ${\rm typeflag}$  & 1 \\
    	$V_{\rm pr}$                & 0.8 pu        & $a$               & 0.8 pu \\
    	$b$                         & 5             & $c$               & 1 s \\
    	\cmidrule{3-4}
    	$d$                         & 0.9 pu        & \multicolumn{2}{|c|}{} \\
    	\hline
    \end{tabular}
    \label{tab3}
    \vspace {-0.5 em}
    \end{table}   
    \begin{align}
    &{\dot x_{\rm D1}} = \frac{1}{{{T_{\rm rv}}}}\left( {{U_{\rm D1}} - {x_{\rm D1}}} \right),\\
	&{\dot x_{\rm D2}} = \frac{1}{{{T_{\rm p}}}}\left( {{x_{\rm D4}} - x_{\rm D2}} \right),\\
	&{\dot x_{\rm D3}} = \frac{1}{{{T_{\rm rf}}}}\left( {U_{\rm D2} - x_{\rm D3}} \right),\\
	&{\dot x_{\rm D4}} = 0 .  
    \end{align}
	To obtain the output power, we also need to calculate the output currents, which are identified as fast states. According to Algorithm 1, there are two options to represent the solutions of fast dynamics depending on the magnitude of $\varepsilon$. For simplicity, it it better to use only the QSS solution to represent the fast states since it does not require solving the boundary-layer model. Let $\varepsilon^{**}=0.06$ to make sure $\max\{\varepsilon_{\rm D}\}<\varepsilon^{**}$, then solving equation (\ref{epsilon**}), we obtain $T=0.242 \;sec$. This means if we use only the QSS solutions, the solution of fast dynamics is inaccurate within $0.242 \;sec$. This time period is intolerably long for stability analysis. Therefore, we should use $z=h+\hat{y}$ by solving the boundary-layer model. The $d$-$q$ axis currents $i_{\rm d}$ and $i_{\rm q}$ are states $S_3$ ($z_{\rm D2}$) and $S_9$ ($z_{D6}$), respectively. Their equations are
	\begin{align}
	&\!\!\!\!\!i_{\rm q}= \sat_2 \left\lbrace \gamma(x_{\rm D}) \right\rbrace  \times  \VP\left(x_{\rm D1},V_{{\mathrm{rfrac}}}\right)+\hat{y}_{\rm D2},\\ 
	&\!\!\!\!\!i_{\rm d}= \sat_{9}\!\left[  \frac{{\sat_7({x_{\rm D4}})}}{{\sat_1\left( x_{\rm D1} \right)}}\right] \times  \VP\left(x_{\rm D1},V_{{\mathrm{rfrac}}} \right)+\hat{y}_{D6}  ,\\
	&\!\!\!\!\!\gamma(x_{\rm D})\!=\!\frac{\tan\!\! \left( {\rm pfaref} \right)\!x_{\rm D2}}{\sat_1\left( x_{\rm D1}\!\right) }\!\!+\!\!K_{\rm qv} \sat_3\left[  \DB_{\rm V}\left( V_{\rm ref0}-x_{\rm D1}\right) \right],
	\end{align}	
	where $\hat{y}_{\rm D2}$ and $\hat{y}_{D6}$ are the solutions of boundary-layer model: 
	\begin{align}
	\dot{y}_{\rm D1}=&-y_{\rm D1},\\ 
	\dot{y}_{\rm D2}=y_{\rm D3}-y_{\rm D2}-& \VP\left(x_{\rm D1},V_{{\mathrm{rfrac}}}\right)\times \nonumber\\
	\left\lbrace \sat_{2}\left[  \gamma(x_{\rm D}) \right]	   + \sat_{2}\right.&\left.\left[  y_{\rm D1} +\gamma(x_{\rm D}) \right] \right\rbrace,\\ 
	\dot{y}_{\rm D3}=&-y_{\rm D3},\\ 
	\dot{y}_{\rm D4}=&-T_{\rm rf}y_{D5},\\ 
	\dot{y}_{D5}=&-y_{D5},\\ 
	\dot{y}_{D6}\!=\!-\sat_{9}\left[\frac{{\sat_7({x_{\rm D4}})}}{{\sat_1\left( x_{\rm D1} \right)}}\right]&-y_{D6} \times  \VP\left(x_{\rm D1},V_{{\mathrm{rfrac}}} \right)    \nonumber\\
	 \!+ \!\sat_{9}\!\left[\!  \frac{{\sat_7({y_{D5}\!+\!x_{\rm D4}})}}{{\sat_1\left( x_{\rm D1} \right)}}\!\right] \!\!\times\! &\left[ y_{\rm D3}\!+\! \VP\left(x_{\rm D1},V_{{\mathrm{rfrac}}} \right)\right].
	\end{align}
\begin{figure}[t!]
	\centering
	\includegraphics[width=8cm]{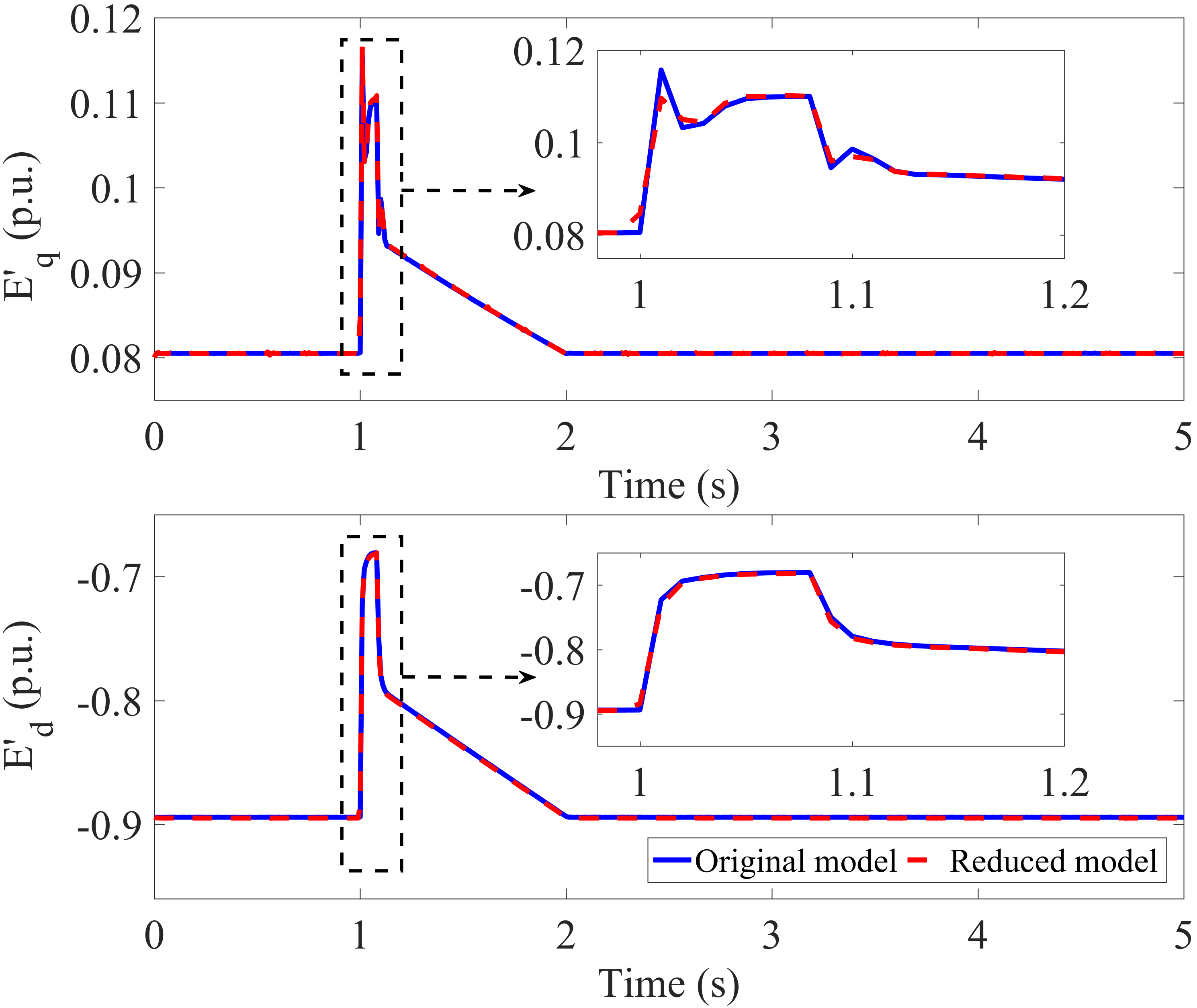}
	\caption{Dynamic responses of $E_{\rm d}'$ and $E_{\rm q}'$ of reduced/original models of three-phase motor A.}
	\label{reduce_Ed_motora}
\end{figure}
\begin{figure}[t!]
	\centering
	\includegraphics[width=8cm]{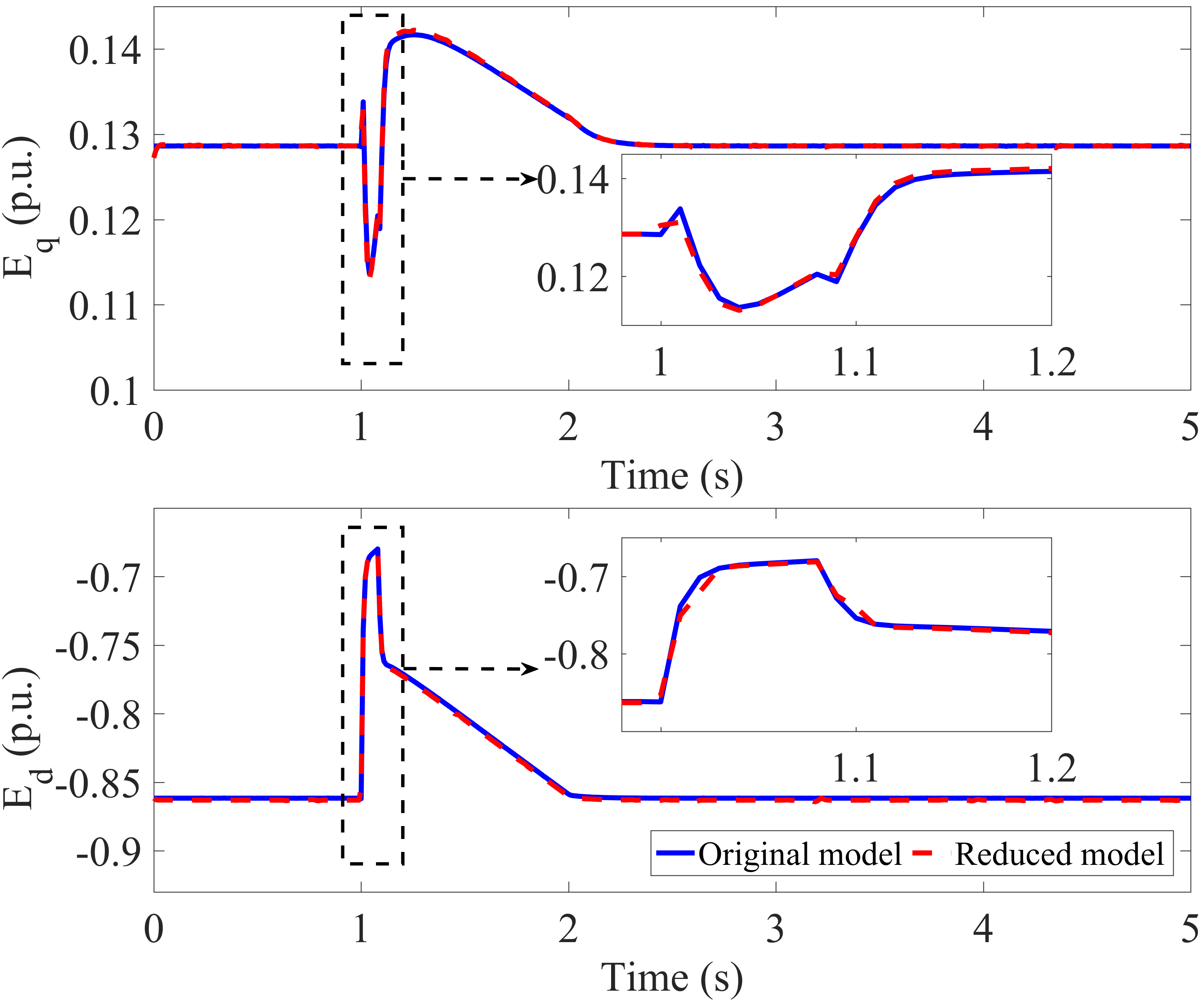}
	\caption{Dynamic responses of $E_{\rm d}'$ and $E_{\rm q}'$ of reduced/original models of three-phase motor B.}
	\label{reduce_Ed_motorb}
\end{figure}
\begin{figure}[t!]
	\centering
	\includegraphics[width=8cm]{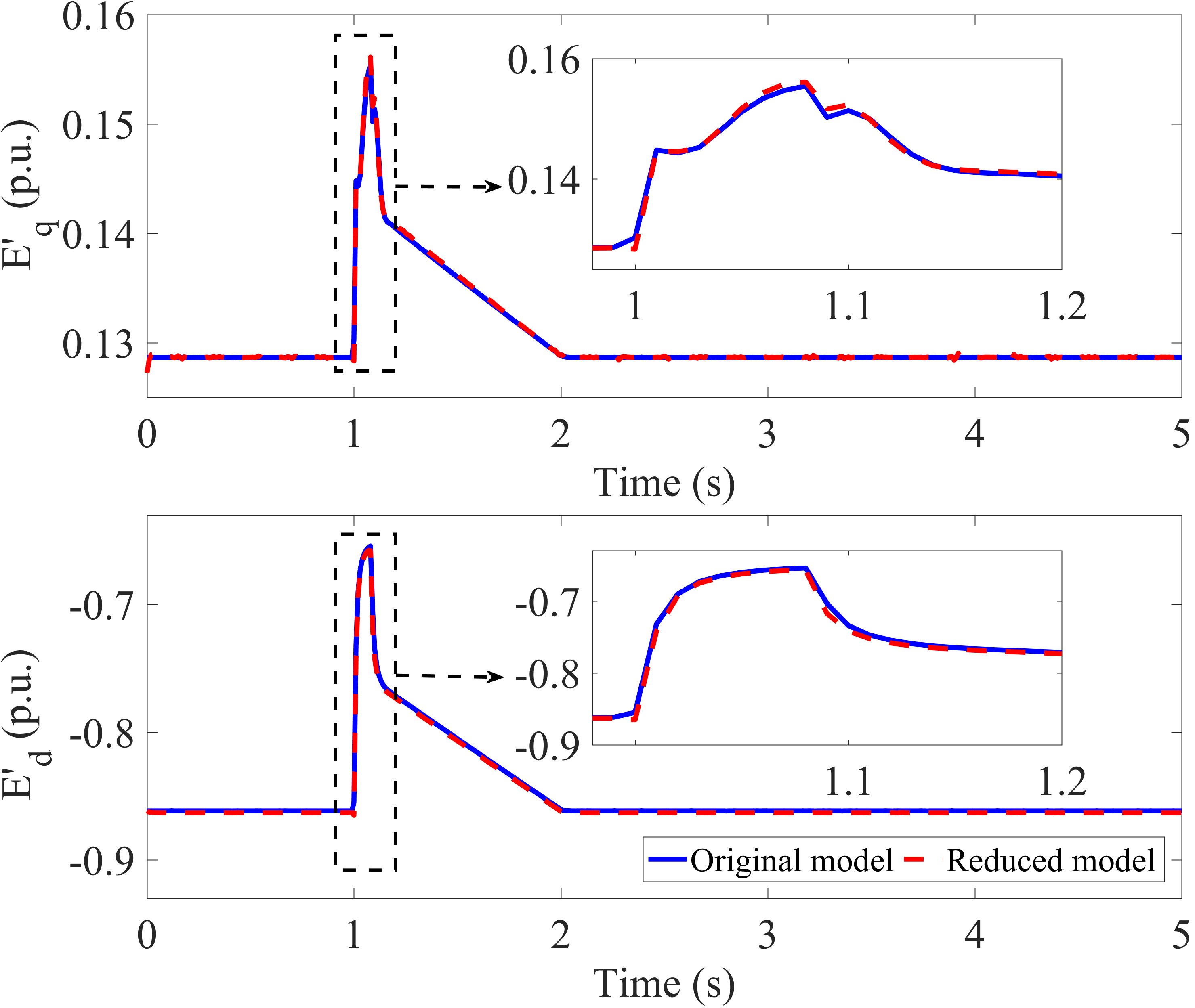}
	\caption{Dynamic responses of $E_{\rm d}'$ and $E_{\rm q}'$ of reduced/original models of three-phase motor C.}
	\label{reduce_Ed_motorc}
\end{figure}

\begin{figure}[t!]
	\centering
	\includegraphics[width=8cm]{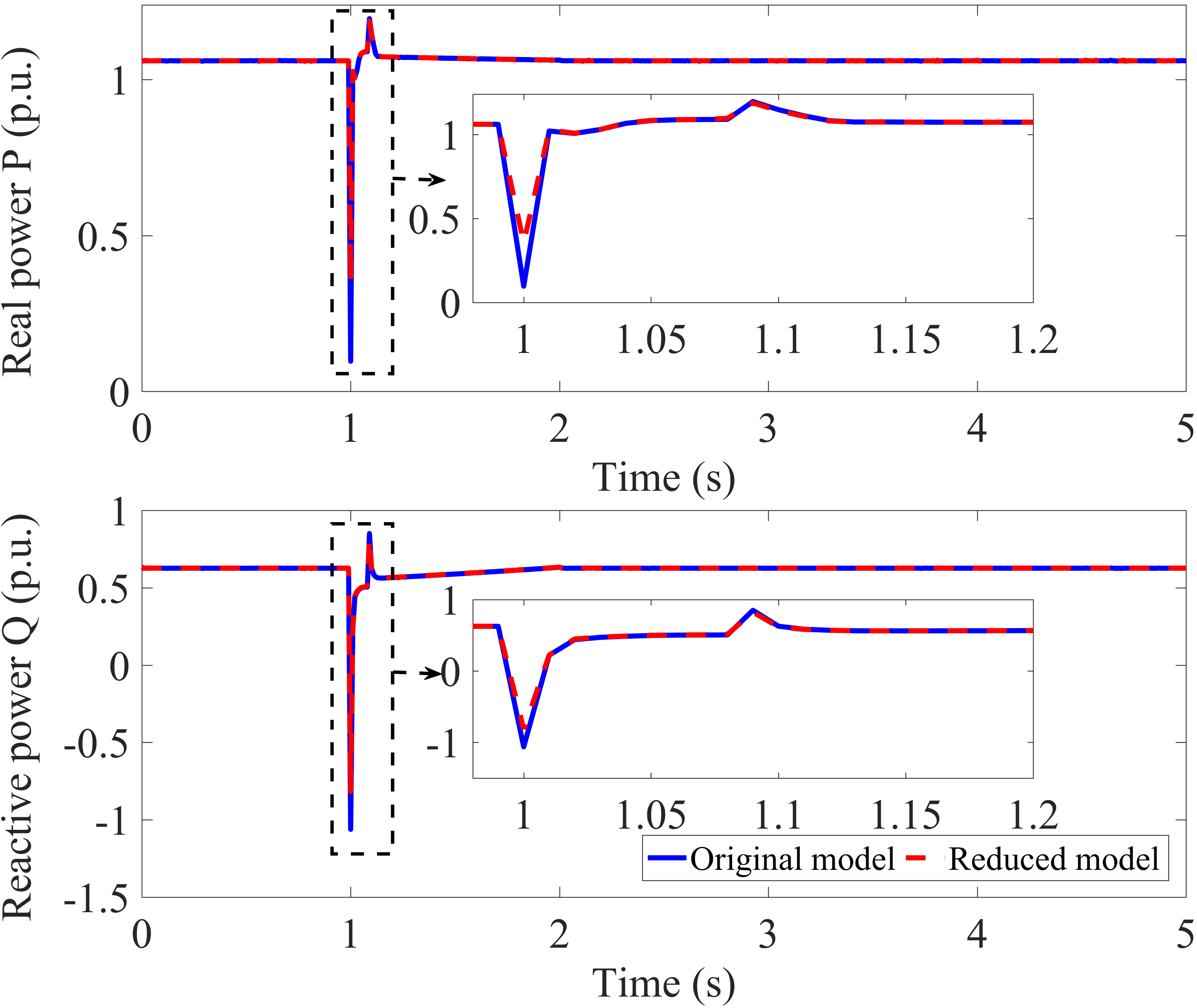}
	\caption{Real/reactive powers of reduced/original models of three-phase motor A.}
	\label{reduce_motora}
\end{figure}
\begin{figure}[t!]
	\centering
	\includegraphics[width=8cm]{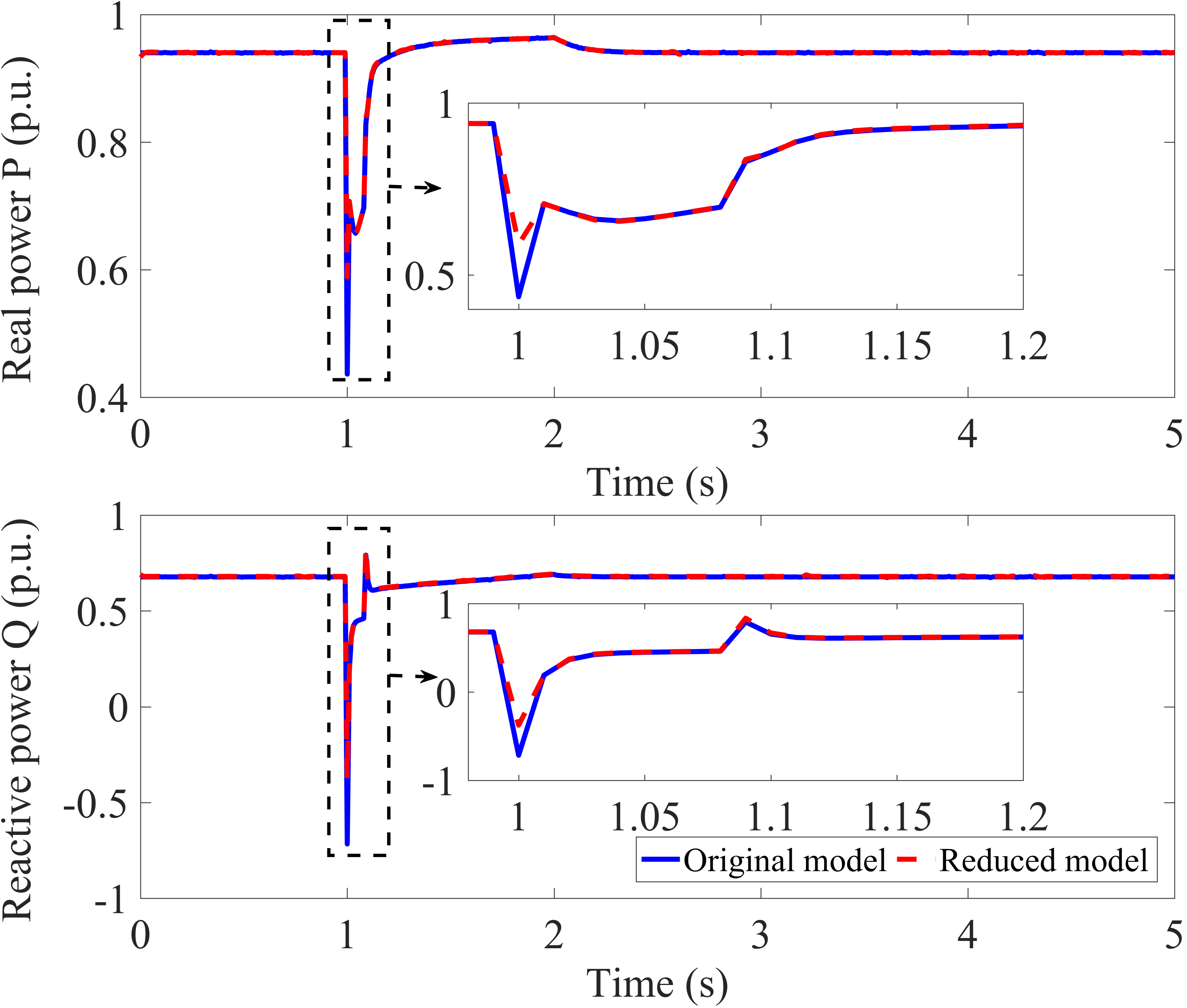}
	\caption{Real/reactive powers of reduced/original models of three-phase motor B.}
	\label{reduce_motorb}
\end{figure}
\begin{figure}[t!]
	\centering
	\includegraphics[width=8cm]{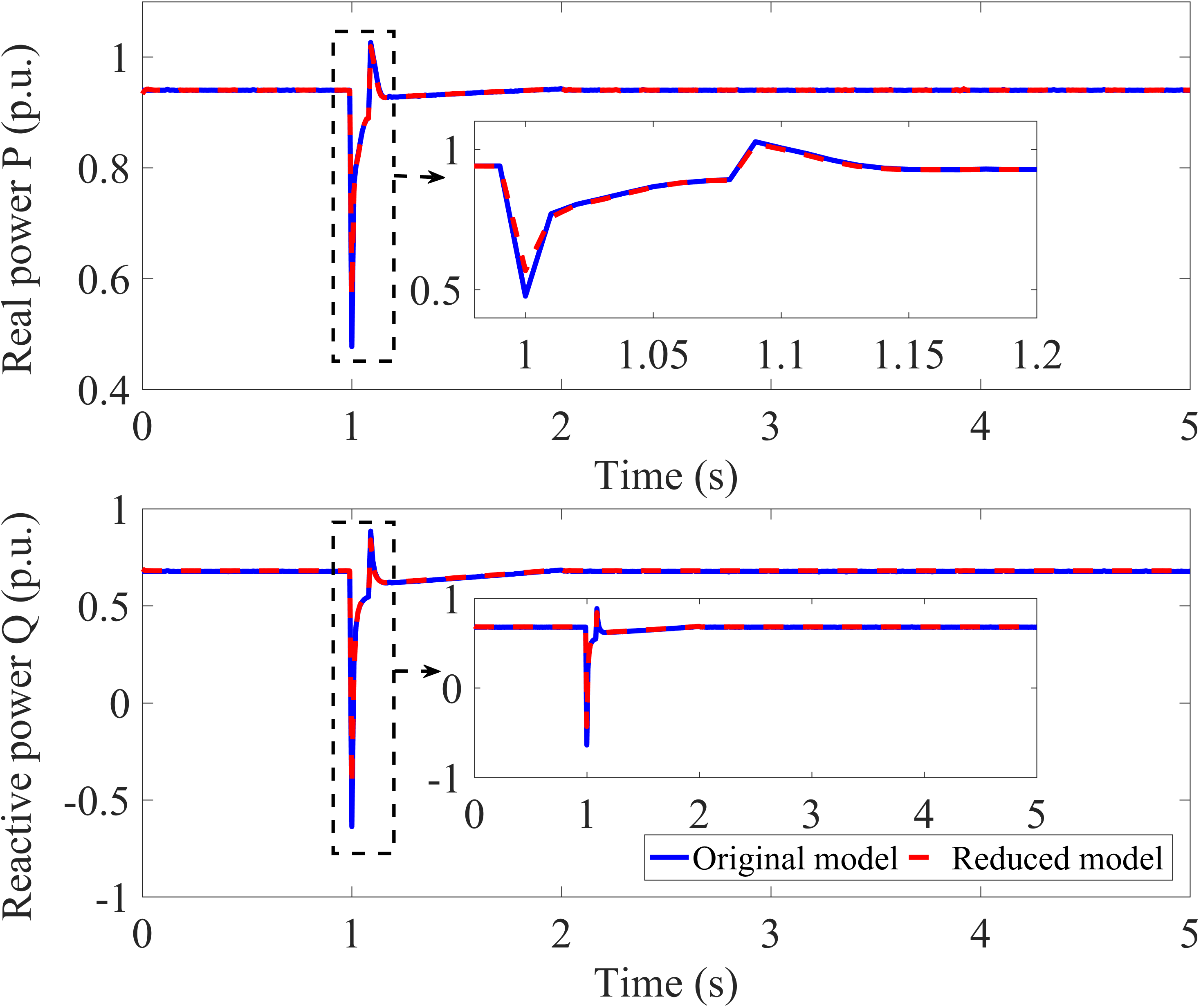}
	\caption{Real/reactive powers of reduced/original models of three-phase motor C.}
	\label{reduce_motorc}
\end{figure}

\section{Model Validation via Simulation}\label{C4}
In this section, the reduced-order models of three-phase motors and DER\_A are tested in Matlab using different solvers. We compare the performance of reduced-order model with original model to verify the effectiveness of the proposed high-fidelity order reduction approach. Moreover, we compare the {computational} time between two models using different solvers to show the reduction of computational burden.  

\subsection{Validation of Reduced-Order Three-Phase Motors}
To verify the proposed reduced-order model of three-phase motor, we {simulate} the reduced and original model in Matlab with the same input voltage. Consequently, we can compare their output power and other states. Refer to an EPRI white paper [21] and [15], this paper tests a voltage sag benchmarking bus voltage input that is generated by (\ref{vgen}). The parameters are set as Table \ref{table2} referring to [22].
\begin{eqnarray}\label{vgen}
V\left( t\right) \! =\! \left\{\! \begin{aligned}
&a{\kern 40pt} \text{if} {\kern 5pt} 1\leqslant t<\left( 1+b/60\right) \!\!\!\!\!\\\!\!\!\!\!\!
&\!\frac{(1\!-\!d)(c\!+\!1\!-t)}{b/60-c}\!+\!1\!{\kern 5pt} \text{if}  {\kern 1pt} \left( 1\!+\!b/60\right) \leqslant \!t\!<\! 1\!+\!c\!\!\!\!\!\\\!\!\!\!\!
&1{\kern 40pt} \text{otherwise}
\end{aligned} \right.
\end{eqnarray}

 Fig. \ref{reduce_Ed_motora}-\ref{reduce_Ed_motorc} show the state responses {of} $E_{\rm q}'$ and $E_{\rm d}'$ {for} three-phase motor A, B and C, respectively. {The blue solid lines denote $E_{\rm q}'$ and $E_{\rm d}'$ of the original model, while the red dashed lines represent those of the reduced-order model. Fig. \ref{reduce_motora}-\ref{reduce_motorc} shows the output real and reactive powers. The blue solid lines denote the real and reactive power of the original model, while the red dashed lines represent those of the reduced-order model.} The mean squared errors of real and reactive power between the original and {reduced-order} model are shown in Table. \ref{table4}. The small errors show the accuracy of the proposed reduced-order three-phase model. Moreover, if using ODE45, which is a solver for non-stiff ODE problems, the {computational} time of the original and {reduced-order} model are 8.8120 {sec} and 0.1926 {sec}, respectively. If using ODE15s, which is a stiff ODE solver, the computational time of the original and reduced model are 1.0975 {sec} and 0.1785 {sec}, respectively. This comparison shows that the singular perturbation method {converts} the original high-order stiff problem to a reduced-order non-stiff problem while {considerably} reducing the {computational} time. This reduction will be more significant in large-scale system with multiple composite loads.
 
\renewcommand\arraystretch{2}
	\begin{table}
		\caption{The mean squared errors between original and reduced-order model of three-phase motor.}
		\label{table4}
		\setlength{\tabcolsep}{2pt}
		\begin{tabular}{|c|c|c|c|}	
			\hline
			\multirow{2}*{\diagbox{Power}{Motor}}&\multicolumn{3}{c|}{Mean Squared Error (MSE)}\\
			\cmidrule{2-4}
			&Motor A&Motor B&Motor C\\
			\hline
			Real power&$1.0509\times10^{-4}$ &$1.1295\times10^{-4}$  &$8.0264\times10^{-5}$ \\
			\hline
			Reactive power&$1.1422\times10^{-5}$ & $1.4294\times10^{-5}$  &		$2.1112\times10^{-5}$ \\
			\hline
		\end{tabular}
		\label{tab4}
		\vspace {-1 em}
	\end{table}

\subsection{Validation of DER\_A Model}
Similar to the verification of three-phase motor, we {simulate} the original {and reduced-order} model of DER\_A in Matlab. The voltage input is the same as (\ref{vgen}). The frequency input is set to be 60 Hz. The parameter setting follows the reliability guideline in [23] as Table \ref{table3}.
\begin{figure}[t!]
	\centering
	\includegraphics[width=7.8cm]{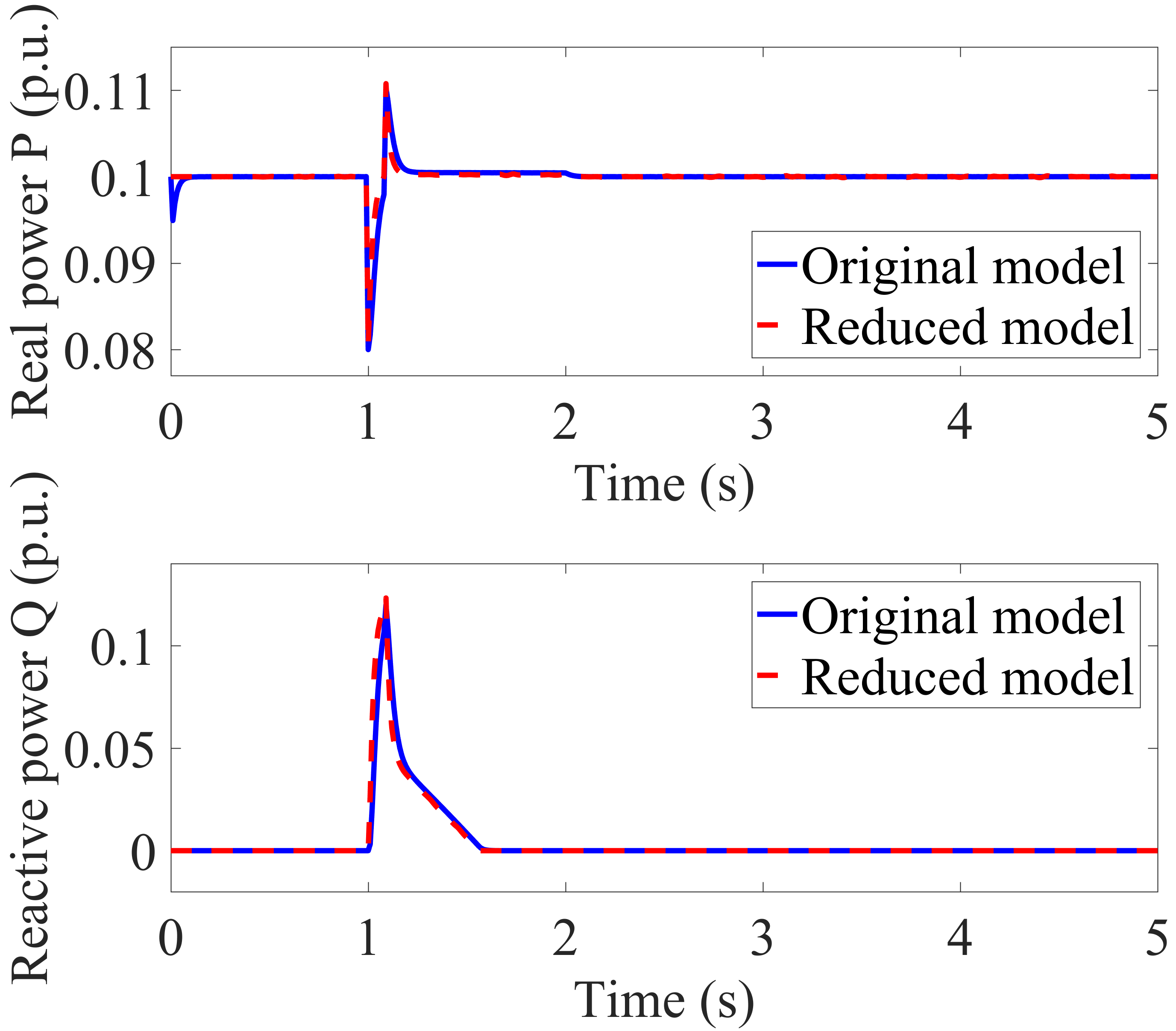}
	\caption{Real and reactive power of reduced and original model of DER\_A.}
	\label{POWER_REDUCE_DERA2}
	\vspace {-1 em}
\end{figure}
\begin{figure}[t!]
	\centering
	\includegraphics[width=7.8cm]{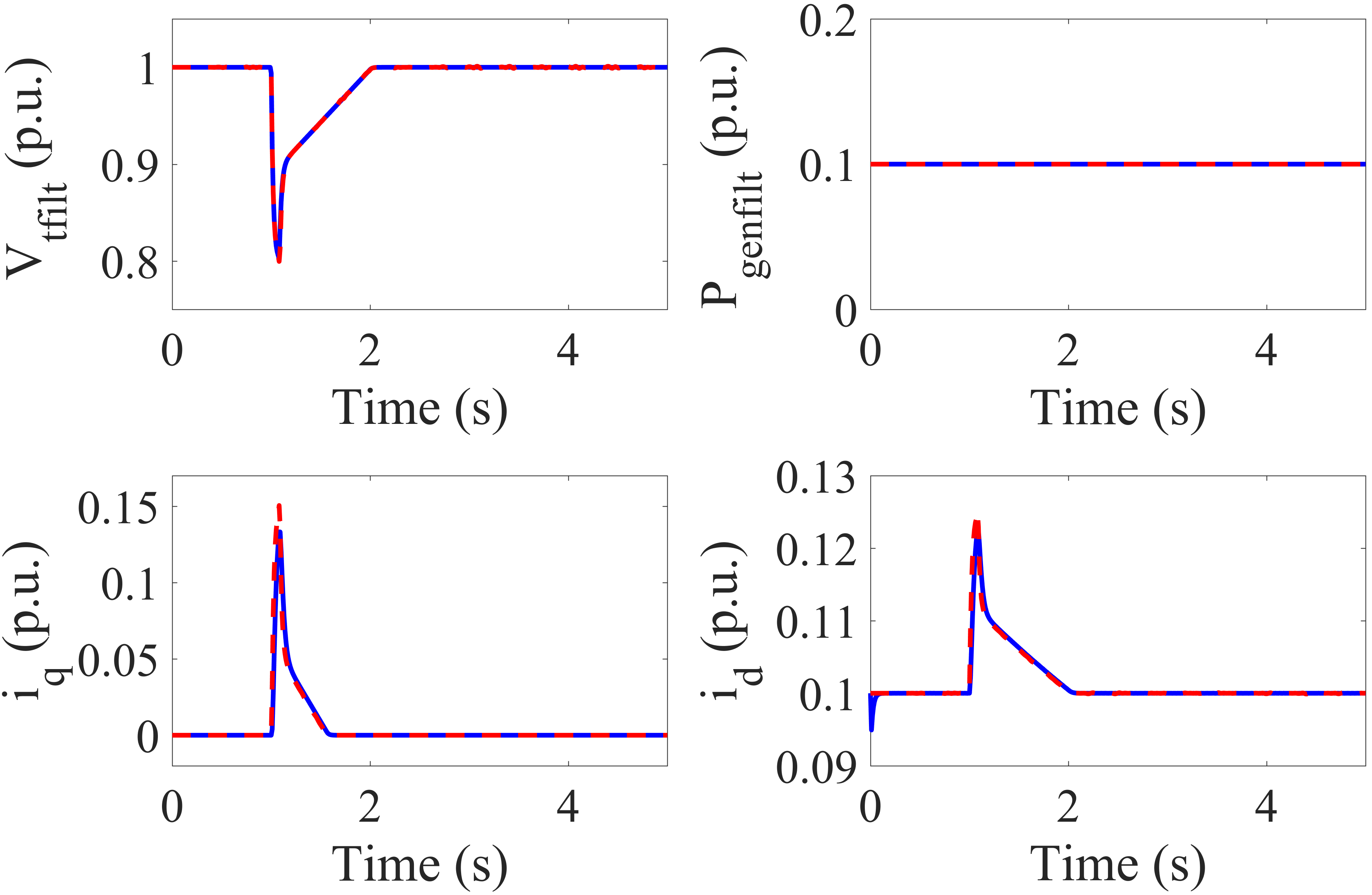}
	\caption{Filtered voltage $V_{\mathrm{tfilt}} $, filtered generated power $P_{\mathrm{genfilt}} $, and filtered current $i_{\rm q}$ and $i_{\rm d}$ of reduced and original model of DER\_A.}
	\label{reduce_dera2}
	\vspace {-1 em}
\end{figure}

Fig. \ref{POWER_REDUCE_DERA2} shows the dynamic responses of DER\_A. The blue lines denote the output powers of original model, while the red lines represent those of reduced one. Fig. \ref{reduce_dera2} shows filtered voltage $V_{\mathrm{tfilt}} $, filtered generated power $P_{\mathrm{genfilt}} $, and filtered current $i_{\rm q}$ and $i_{\rm d}$ of reduced and original model of DER\_A. The mean square errors (MSE) of real and reactive power are $7.1363\times10^{-4}$ and $1.3045\times10^{-5}$, respectively. Further, the {computational} time of original and {reduced-order} model using ODE45 are 11.205 {sec} and 0.2074 {sec}, respectively; the {computational} time of original and {reduced-order} model using ODE15s are 2.0012 {sec} and 0.1598 {sec}, respectively. 
	
\section{Conclusion}\label{C5}
This paper proposes a high-fidelity large-signal order reduction approach for the latest WECC composite load model including DER\_A. The derived reduced-order model has guaranteed high accuracy that can replace the original load model in high-order system simulation to perform stability analysis, optimization and controller design. This replacement can significantly reduce the difficulty of stability analysis and computational burden. The simulation results verify the accuracy and efficiency of the proposed algorithm.

\section{References}\label{sec13}
\begin{enumerate}
\item[{[1]}] C. W. Taylor: 'Power system voltage stability' (New York: McGraw-Hill, 1994.)\vspace*{6pt}

\item[{[2]}] C.  Wang,  Z.  Wang,  J.  Wang,  and  D.  Zhao: 'SVM-based  parameter identification for composite ZIP and electronic load modeling', IEEE Trans. Power Syst., 2019, 34, (1), pp. 182–193\vspace*{6pt}

\item[{[3]}] J.  Zhao,  Z.  Wang,  and  J.  Wang: 'Robust  time-varying  load  modeling for  conservation  voltage  reduction  assessment', IEEE Trans. Smart Grid, 2018, 9, (4), pp. 3304–3312\vspace*{6pt}

\item[{[4]}] D. Kosterev, A. Meklin, J. Undrill, B. Lesieutre, W. Price, D. Chassin,R. Bravo, and S. Yang: 'Load modeling in power system studies: WECC progress  update',  Proc.  IEEE  Power  and  Energy  Society  General Meeting - Conversion  and  Delivery  of  Electrical  Energy  in  the  21st Century, 2008, pp. 1–8\vspace*{6pt}

\item[{[5]}] C. Wang, Z. Wang, J. Wang, and D. Zhao: 'Robust time-varying parameter identification for composite load modeling', IEEE Trans. Smart Grid, 2019, 10, (1), pp. 967–979\vspace*{6pt}

\item[{[6]}] J. Wang, W. Zuo, L. Rhode-Barbarigos, X. Lu, J. Wang, and Y. Lin: 'Literature review on modeling and simulation of energy infrastructures from a resilience perspective', Reliability Engineering \& System Safety, 2018\vspace*{6pt}

\item[{[7]}] K. Zhang, H. Zhu, and S. Guo: 'Dependency  analysis  and  improved parameter  estimation  for  dynamic  composite  load  modeling', IEEE Trans. Power Syst., 2016, 32, (4), pp. 3287–3297\vspace*{6pt}

\item[{[8]}] S.  Guo,  K.  S.  Shetye,  T.  J.  Overbye,  and  H.  Zhu: 'Impact  of  measurement selection on load model parameter estimation,” 2017 IEEE Power and Energy Conference at Illinois (PECI).  IEEE, 2017, pp. 1–6.\vspace*{6pt}

\item[{[9]}] M. Cui, J. Wang, Y. Wang, R. Diao, and D. Shi: 'Robust time-varying synthesis  load  modeling  in  distribution  networks  considering  voltage disturbances', IEEE Trans. Power Syst., pp. 1–1, 2019.\vspace*{6pt}

\item[{[10]}] C. Fu, Z. Yu, D. Shi, H. Li, C. Wang, Z. Wang, and J. Li: 'Bayesian estimation based parameter estimation for composite load', arXiv preprint 2019, arXiv: 1903.10695.\vspace*{6pt}

\item[{[11]}] D. N. Kosterev, C. W. Taylor, and W. A. Mittelstadt: 'Model validation for the August 10, 1996 WSCC system outage', IEEE Trans. Power Syst., 1999, 14, (3), pp. 967–979.\vspace*{6pt}

\item[{[12]}] A. Arif, Z. Wang, J. Wang, B. Mather, H. Bashualdo, and D. Zhao: 'Load Modeling—A Review', IEEE Trans. Smart Grid, 2018, 9, (6), pp. 5986–5999\vspace*{6pt}

\item[{[13]}] W. E. C. Council: 'WECC dynamic composite load model (CMPLDW) specifications', 2015.\vspace*{6pt}

\item[{[14]}] Q. Huang, R. Huang, B. J. Palmer, Y. Liu, S. Jin, R. Diao, Y. Chen, and Y.  Zhang: 'A  generic  modeling  and  development  approach  for  WECC composite load model', Electr. Power Syst. Res., 2019,  172, pp.1–10.\vspace*{6pt}

\item[{[15]}] Electrical  Power  Research  Institute  (EPRI): 'The  new  aggregated distributed energy resources (DER\_A) model for transmission planning studies: 2019 update', 2019.\vspace*{6pt}

\item[{[16]}] H. K. Khalil: 'Nonlinear Systems'. New Jersey: Prentice Hall, 2000. \vspace*{6pt}

\item[{[17]}] S.  D.  Pekarek,  M.  T.  Lemanski,  and  E.  A.  Walters,  “On  the  use  of singular  perturbations  to  neglect  the  dynamic  saliency  of  synchronous machines', IEEE Trans. Energy Convers, 2002, 17, (3), pp.385–391.\vspace*{6pt}

\item[{[18]}] M.  Rasheduzzaman,  J.  A.  Mueller,  and  J.  W.  Kimball: 'Reduced-order  small-signal  model  of  microgrid  systems', IEEE Trans. Sustain. Energy, 2015, 6, (4), pp. 1292–1305. \vspace*{6pt}

\item[{[19]}] R.  M.  G.  Castro  and  J.  M.  Ferreira  de  Jesus: 'A  wind  park  reduced-order model using singular perturbations theory', IEEE Trans. Energy Convers, 1996, 11, (4), pp. 735–741. \vspace*{6pt}

\item[{[20]}] Z. Ma, Z. Wang, Y. Wang, R. Diao, and D. Shi: 'Mathematical representation  of  the  WECC  composite  load  model', arXiv  preprint, 2019, arXiv:1902.08866.\vspace*{6pt}

\item[{[21]}] Electric Power Research Institute (EPRI): 'Distributed energy resources modeling for transmission planning studies. Summary modeling guidelines', 2016.\vspace*{6pt}

\item[{[22]}] X. Wang, Y. Wang, D. Shi, J. Wang and Z. Wang: 'Two-stage WECC composite load modeling: a double deep q-Learning networks approach', IEEE Trans. Smart Grid, 2020, in press.\vspace*{6pt}

\item[{[23]}] North American Electric Reliability Corporation (NERC): 'Reliability guideline-parameterization of the DER\_A model', June, 2019.\vspace*{6pt}
\end{enumerate}


\end{document}